\documentclass[11pt,reqno,a4paper]{article}
\usepackage{amsmath}
\usepackage{amsbsy}
\usepackage{amsthm}
\usepackage{amssymb}
\usepackage{amscd}
\usepackage{cancel}
\usepackage{accents}
\usepackage{float}
\usepackage{url}
\usepackage{tikz}
\usetikzlibrary{matrix,arrows,decorations.pathmorphing}
 \usepackage{relsize}
\usepackage{xcolor, mathrsfs}

\usepackage{comment}
\usepackage{enumerate}

\newcommand{\rquot}[2]{\raisebox{0.5ex}{$#1$}\!/\!\raisebox{-0.5ex}{$#2$}}

\usepackage{subcaption}

\usepackage{hyperref}
\hypersetup{
pdftitle={},%
pdfauthor={},%
pdfsubject={},%
pdfkeywords={},%
colorlinks=true,%
linkcolor={black},%
linktoc={},
linktocpage={false},%
pageanchor={},
citecolor={black},
}

\usepackage[english]{babel}
\usepackage[utf8]{inputenc}
\usepackage[margin=2.41cm]{geometry}

\usepackage{comment}
\usepackage{ulem}

\usepackage{cite}

\usepackage{titlesec}
\titleformat{\section}{\large\bfseries\filcenter}{\thesection}{1em}{}
\titleformat{\subsection}{\bfseries}{\thesubsection}{1em}{}

\input{xy}
\xyoption{all}

\newtheorem{theorem}{Theorem}[section]
\newtheorem{corollary}[theorem]{Corollary}
\newtheorem{lemma}[theorem]{Lemma}
\newtheorem{proposition}[theorem]{Proposition}
\theoremstyle{remark}
\theoremstyle{remarks}

\theoremstyle{definition}
\newtheorem{remark}[theorem]{Remark}
\newtheorem{remarks}[theorem]{Remarks}
\newtheorem{definition}[theorem]{Definition}
\newtheorem{example}[theorem]{Example}

\newtheorem{setup}[theorem]{Setup}

\numberwithin{equation}{section}
\allowdisplaybreaks[1]

\catcode`@=12

\renewcommand\thanks[1]{%
  \begingroup
  \renewcommand\thefootnote{}\footnote{#1}%
  \addtocounter{footnote}{-1}%
  \endgroup
}

\renewcommand{\tilde}{\widetilde}
\renewcommand{\epsilon}{{\varepsilon}}

\usepackage{bbm}

\newcommand{\ie}{\textit{i.e. }}
\newcommand{\cf}{\textit{cf. }}

\newcommand{\KK}{{\mathbb{K}}}
\newcommand{\RR}{{\mathbb{R}}}
\newcommand{\CC}{{\mathbb{C}}}
\newcommand{\n}{{\mathtt{n}}}

\renewcommand{\hat}{\widehat}

\newcommand{\Dir}{{\mathsf{D}}}
\newcommand{\G}{\mathsf{G}}
\newcommand{\Id}{{\operatorname{Id}}}
\newcommand{\sol}{\mathsf{Sol}\,}

\newcommand{\E}{\mathsf{E}}

\newcommand{\I}{\mathsf{I}}
\renewcommand{\L}{\mathsf{L}}
\newcommand{\M}{\mathsf{M}}

\renewcommand{\H}{\mathsf{H}}

\renewcommand{\P}{\mathsf{P}}

\newcommand{\R}{\mathsf{R}}
\renewcommand{\S}{\mathbb{S}}
\newcommand{\T}{\mathsf{T}}

\newcommand{\Z}{\mathsf{Z}}

\newcommand{\bM}{{\partial\M}}
\newcommand{\bSigma}{{\partial\Sigma}}

\newcommand{\oS}{\mathsf{S}}

\newcommand{\oB}{\mathsf{B}}
\newcommand{\f}{{\mathfrak{f}}}

\newcommand{\h}{{\mathfrak{h}}}

\def\bkappa{\kappa^f}

\newcommand{\fR}{\mathfrak{R}}

\newcommand{\fA}{{\mathfrak{A}}}

\renewcommand{\aa}{{\mathfrak{a}}}

\newcommand{\vol}{{\textnormal{vol}\,}}
\newcommand{\End}{{\textnormal{End}\,}}
\newcommand{\supp}{{\textnormal{supp\,}}}
\newcommand{\Spin}{\textnormal{Spin}}

\newcommand{\fiber}[2]{\prec  #1\,|\, #2  \succ}

\newcommand{\scalar}[2]{(#1\, |\, #2)}

\definecolor{NiColor}{RGB}{0,128,0}

\begin{document}

\begin{flushright}

\baselineskip=4pt

\end{flushright}

\begin{center}
\vspace{5mm}

{\Large\bf  M\O LLER OPERATORS AND HADAMARD STATES\\[4mm]
FOR DIRAC FIELDS WITH MIT BOUNDARY CONDITIONS}

\vspace{5mm}

{\bf by}

\vspace{5mm}

{  \bf  Nicol\`o Drago}\\[1mm]
\noindent  {\it Dipartimento di Matematica, Universit\`a di Trento, 38050 Povo (TN), Italy}\\[1mm]
email: \ {\tt  nicolo.drago@unitn.it}\\

\vspace{5mm}
{  \bf  Nicolas Ginoux}\\[1mm]
\noindent  {\it Universit\'e de Lorraine, CNRS, IECL, F-57000 Metz, France}\\[1mm]
email: \ {\tt  nicolas.ginoux@univ-lorraine.fr}\\

\vspace{5mm}

{  \bf  Simone Murro}\\[1mm]
\noindent  {\it Laboratoire de Mathématiques d’Orsay, Université Paris-Saclay, 91405 Orsay, France}\\[1mm]
email: \ {\tt  simone.murro@u-psud.fr}
\\[10mm]
\end{center}

\begin{abstract}
The aim of this paper is to prove the existence of Hadamard states for Dirac fields coupled with MIT boundary conditions  on any globally hyperbolic manifold with timelike boundary.
 This is achieved by introducing a geometric
M\o ller operator which implements a unitary isomorphism between the spaces 
of $L^2$-initial data 
of particular symmetric systems we call weakly-hyperbolic and which are coupled 
with admissible boundary 
conditions.
In particular, we show that for Dirac fields with MIT boundary 
conditions, this isomorphism can be lifted to a $*$-isomorphism between the 
algebras of Dirac fields and that any Hadamard state can be pulled back along this 
$*$-isomorphism preserving the singular structure of its two-point 
distribution. 
\end{abstract}

\\[0.5mm]

\tableofcontents

\renewcommand{\thefootnote}{\arabic{footnote}}
\setcounter{footnote}{0}

\section{Introduction}
The initial value problem for a symmetric hyperbolic system on a
Lorentzian manifold $\M$ is a classical problem which has been exhaustively
studied in many contexts. 
If the underlying background is \textit{globally hyperbolic}, a complete answer 
is known: In~\cite{BaGreen} it has been shown that the Cauchy problem is 
well-posed for any smooth initial data.
Even if there exists a plethora of models in physics where globally hyperbolic
spacetimes have been used as a background, there also exist many applications
which require a manifold with non-empty boundary. Indeed, recent developments in quantum field theory focused their attention on manifolds with timelike boundary~\cite{aqftCAT,Zahn}, e.g. anti-de Sitter spacetime~\cite{Dappia4,DappiaMarta} and BTZ spacetime~\cite{BTZ}. Moreover, experimental setups for studying the Casimir effect enclose
(quantum) fields between walls, which may be mathematically described by
introducing timelike boundaries~\cite{cospi,Casimir}. Also moving walls in a spatial set-up
correspond to a timelike boundary in the Lorentzian manifold. For the class of  
globally hyperbolic manifolds  with timelike boundary, the Cauchy problem was 
investigated by the last two named authors. 
In particular, we showed  
in~\cite{Ginoux-Murro-20} 
that the Cauchy problem for any symmetric hyperbolic system  coupled with an 
admissible boundary condition is well-posed and the unique solution propagates 
with at most the speed of light.
As a byproduct the existence of a causal propagator is guaranteed.
This operator plays a pivotal role in the algebraic approach to quantum field theory since it allows to construct an algebra of 
observables in a covariant manner, see e.g.~\cite{gerard, aqft2} for textbooks, \cite{CQF1,FK} for recent reviews and~\cite{Nic,cap,simo1,simo2,simo4,DHP,DMP3, CSth,cospi,Dappia4,SDH12} for some applications.
In order to complete the quantization of a free field theory, it is necessary to define an 
(algebraic) state, \ie a positive functional on the algebra of the 
observables.
Clearly not any state can be considered physically meaningful and, 
on globally hyperbolic spacetimes with empty boundary, only those satisfying 
the so-called Hadamard condition are regarded as states of physical interest.

Indeed, within this setting, Hadamard states guarantee the possibility of 
constructing Wick polynomials following a local and covariant scheme 
\cite{HW,IgorValter}, granting also the finiteness of the quantum fluctuations 
of such Wick polynomials, see e.g.~\cite{Fewster-Verch-13}.
Let us remark that in globally hyperbolic stationary spacetimes with empty 
boundary, the ground state and any KMS state satisfy the Hadamard condition, see 
e.g.~\cite{HadGroundKMS1,HadGroundKMS2}.
In close analogy, in the presence of a timelike boundary a generalization of the Hadamard condition has been proposed in \cite{Michal}.
Once the Hadamard condition is assumed, several natural questions arise. The 
most important one concerns the existence of such states, a problem which was 
answered positively for free field theories on globally hyperbolic spacetimes 
with empty boundary in~\cite{FNW} (except linearised gravity for which the 
method cannot be applied -- see for 
example \cite{simogravity}) by means of a spacetime deformation 
technique. 
In detail, once a globally hyperbolic manifold $(\M,g)$ is 
assigned, the key point is to find a (ultra-)static globally hyperbolic metric 
$g_0$ on $\M$ has well as a Lorentzian metric $g_\chi$ interpolating between 
$g$ and $g_0$ which is globally hyperbolic. This is not an easy task, since the 
convex combination of two given globally hyperbolic metric is not in general 
globally hyperbolic. 
If the boundary is not empty, the situation gets worse.
This is first due to the need for a 
boundary condition for Cauchy problem.
Furthermore, if the interpolating metric 
is not explicit, then it is not clear how to construct that boundary condition.
Hence the arguments used in~\cite{FNW} cannot be applied in a straightforward 
manner and a new proof has to be thought out. 
\medskip

The aim of this paper is to provide a geometric proof of the existence of 
Hadamard states for Dirac field with a boundary condition dubbed \textit{MIT 
boundary condition}. Let us recall that the {MIT boundary condition} is a local 
boundary condition which was introduced for the first time in \cite{MIT1} in 
order to reproduce the confinement of quark in a finite region of space: 
``Dirac waves'' are indeed reflected on the boundary, see also \cite{MIT2,MIT3} 
for the description of hadronic states, like baryons and mesons. The MIT 
boundary condition has been used more recently for many other applications, like 
the computation of the Casimir energy in a three-dimensional rectangular 
box~\cite{Casimir,FR1,FR2} in order to construct an integral 
representation for the Dirac propagator in Kerr Black Hole Geometry and finally 
also in \cite{IR} to prove the asymptotic completeness for linear massive 
Dirac fields on the Schwarzschild Anti-de Sitter spacetime.
\
A summary of our main result is the following:

\begin{theorem}\label{thm:Hadapl}
Let $(\M,g:=-\beta^2dt^2 + h(t))$ be a globally hyperbolic spin spacetime with 
timelike boundary and let $\Dir$ be the Dirac operator coupled with MIT boundary 
condition ---\textit{cf.} Equation \eqref{eq: MIT bc}.
If 
for any $u\in\sol_{\textsc{mit}}(\Dir)$,  the $b$-wave front set $\operatorname{WF}_b(u)$ is the union of maximally extended generalized broken bicharacteristics, then there exists a  state for the 
algebra of Dirac fields with MIT boundary conditions that satisfies the 
Hadamard condition as per Definition \ref{def: Hadamard state}.
\end{theorem}

\begin{remark}
\noindent
The requirement in Theorem~\ref{thm:Hadapl} is also known as ``propagation of 
singularity theorem'' and it has been used for scalar wave equation in many 
different settings \textit{e.g.} \cite{DappiaMarta, GannotWrochna, 
Melrose-Sjostrand-78,Melrose-Sjostrand-82,Vasy-08,mosterbook,Taylor-75,
Baskin-Wunsch-20}. 
We expect a similar result to hold since the propagation of 
singularity for Dirac operators reduces to the propagation of singularity for 
scalar wave operator. 
This is because any element in the kernel of the Dirac 
operator is also in the kernel of the spinorial wave operator, whose principal 
symbol is equal to the principal symbol of the scalar wave equation times an 
identity matrix. 
What is left to check is how the boundary condition for 
the Dirac operator and the spinorial wave equation are related. 
In globally hyperbolic asymptotically anti-de Sitter spacetimes
\footnote{A globally hyperbolic manifold $(\M,g)$ is called asymptotically anti-de Sitter spacetime if it holds the following: (1) for any boundary function $x$ the metric $\hat g= x^2 g$ extends smoothly to a Lorentzian metric on $\M$; (2) the pullback $\iota^*_{\bM} \hat g$ is a smooth Lorentzian metric on $\bM$; (3) $\hat g^\sharp(dx,dx)=1$ on $\bM$.}, Dappiaggi and Marta 
in~\cite{DappiaMarta} proved the propagation of singularity for the scalar wave 
equation for a very large class of boundary conditions, which contains all 
self-adjoint boundary conditions. 
For these reasons, we expect that the Dirac 
operator coupled with MIT boundary condition should also enjoy a propagation of 
singularity in this class of spacetimes.
\end{remark}
Our strategy to prove the existence of Hadamard states is as follows. 
In Section~\ref{sec:SWHS} we introduce the class of symmetric weakly-hyperbolic 
operators (\cf Definition~\ref{def:weakly-hyper}) extending that of 
symmetric hyperbolic ones.
We show that the Cauchy problem is well-posed for that generalized category of 
operators (\cf Theorem~\ref{thm:main1}). 
Then in Section~\ref{sec:Moller} we construct a 
M\o ller operator, \ie a geometric map which compares the space of solutions 
of two given symmetric weakly-hyperbolic systems coupled with admissible 
boundary conditions on (possibly different) globally hyperbolic manifolds with 
timelike boundary (\cf Theorem~\ref{thm:Moller}).
In particular 
Section~\ref{sec:conserv}, we show that this geometric map can be constructed to 
preserve the natural scalar product defined on the space of solution (\cf 
Proposition~\ref{prop:conserv scal prod}). In Section~\ref{sec:AQFT} we 
specialize ourself to the case of Dirac operators: after introducing the 
classical Dirac operator and the MIT boundary conditions~\ref{sec:Dirac} 
and~\ref{sec:MIT} respectively, we construct an isomorphism between  spinor 
bundles defined on different Lorentzian manifolds, see Section~\ref{sec:kappa}.
In 
Section~\ref{sec:CAR} we construct the algebras of Dirac fields and we promote 
the unitary map between the spaces of solutions of the Dirac equation to a 
$*$-isomorphism between the algebras of Dirac fields (\cf Theorem~\ref{thm:alg 
iso}). 
Finally in Section~\ref{sec:Hadam} we discuss and prove the existence of 
Hadamard states for Dirac fields with MIT boundary conditions.

\subsection*{Acknowledgements}
We would like to thank Matteo Capoferri, Claudio Dappiaggi, Christian G\'erard, Valter Moretti, Miguel S\'anchez, Daniele Volpe and Micha\l\,  Wrochna for helpful discussions related to the topic of this paper.

\paragraph{Funding} 
S.M. is supported by the DFG research grant MU 4559/1-1 ``Hadamard States in Linearized Quantum Gravity''.

\subsection*{Notation and convention}
\begin{itemize}
\item[-] The symbol $\KK$ denotes one of the elements of the set $\{\RR,\CC\}$.
\item[-] $\M:=(\M,g)$ is a globally hyperbolic manifold  with timelike boundary 
$\bM$ and we adopt the convention that $g$ has the signature $(-,+,\dots,+)$.
If $g$ is a Lorentzian metric such that $(\M,g)$ is globally hyperbolic, then we shall write $g\in\mathcal{GH}_\M$.
\item[-] For two Lorentzian metrics $g,g'$, $g\leq g'$ means that any causal tangent vector for $g$ is causal for $g'$
or equivalently $J_{g'} \subset J_g$.
\item[-] $t:\M\to\RR$ is a Cauchy temporal function and $\M_\T:=t^{-1}(t_0,t_1)$ 
is a time strip.
\item[-] $\n$ is the outward unit normal vector to $\bM$.
\item[-] $\flat:\T\M\to \T^*\M$ and $\sharp:\T^*\M\to \T\M$ are the musical 
isomorphisms.
\item[-] $\E$ is a $\KK$-vector bundle over $\M$ with $N$-dimensional fibers, 
denoted by $\E_p$ for $p\in\M$, and endowed with a Hermitian fiber metric 
$\fiber{\cdot}{\cdot}_p$ .
\item[-]  $\Gamma_{c}(\E), \Gamma_{sc}(\E)$ \textit{resp.} $\Gamma(\E)$ denote the 
spaces of compactly supported, spacelike compactly supported \textit{resp.} smooth 
sections of $\E$.
\item[-] $\oS$ is a symmetric weakly-hyperbolic system of constant 
characteristic coupled with principal symbol denoted by $\sigma_\oS$ and $\oB$ 
is admissible boundary space for $\oS$.
\item[-] When $\M$ is a Lorentzian spin manifold, we denote with $\S\M$ the spinor bundle over $\M$ and with $\Dir$ the classical Dirac operator.
\end{itemize}

\section{M\o ller operators for symmetric weakly-hyperbolic 
systems}\label{sec:preliminaries}

The aim of this section is to construct a geometric map, named M\o ller 
operator, to compare the solution 
spaces of two symmetric weakly-hyperbolic operators coupled with admissible 
boundary conditions on possibly different (though related) globally hyperbolic manifolds with 
timelike boundary.
To this end, we shall first recall the theory of symmetric 
hyperbolic systems on globally hyperbolic manifolds with timelike boundary. 
Then, after showing the well-posedness of the Cauchy problem for weakly hyperbolic systems, we shall construct a family of M\o ller 
operators depending on the choice of an arbitrary smooth function $f$.
By setting suitably such a function, we shall show that the resulting 
M\o ller operator is actually a unitary map between the spaces of initial data 
endowed with a naturally defined positive scalar product.
Our goal is achieved with the help of~\cite{Ginoux-Murro-20, defarg}.\medskip 

\subsection{Globally hyperbolic manifolds}
Let $\M$ be a connected oriented smooth manifold with boundary.
We assume $\M$ to be endowed with a smooth Lorentzian metric $g$ for which $\M$ 
becomes time-oriented. 
Here and 
in the following we shall assume that the boundary is timelike, \ie
the pullback of $g$ with respect to the natural inclusion $\iota\colon \bM \to 
\M$ 
defines a 
Lorentzian metric $\iota^*g$ on the boundary.
In the class of Lorentzian 
manifolds with timelike boundary, those called globally hyperbolic provide a 
suitable background where to analyze the Cauchy problem  for hyperbolic 
operators. 
\begin{definition}{\protect{\cite[Definition 2.14]{Ake-Flores-Sanchez-18}}}\label{def:glob hyper manifold timelike boundary}
A \textit{globally hyperbolic manifold with timelike boundary} is an $(n + 
1)$-dimensional, oriented,
time-oriented, smooth Lorentzian manifold $\M$ with timelike boundary $\bM$ such 
that
\begin{enumerate}[(i)]
\item $\M$ is causal, \ie  there are no closed causal curves;
\item for all points $p,q\in\M$, the subset $J^+(p)\cap J^-(q)$ of $\M$ is 
compact, where 
$J^+(p)$ (\textit{resp.} $J^-(p)$) denotes the causal future (\textit{resp}. 
past) of $p$ (\textit{resp.} q) in $\M$.
\end{enumerate}
\end{definition}
\begin{remark}
 In case of an empty boundary, this definition agrees with the standard one, 
see e.g.~\cite[Section 3.2]{Bee} or \cite[Section 1.3]{wave}. 
\end{remark}
Recently, Ak\'e, Flores and S\'anchez gave a characterization of globally 
hyperbolic 
manifolds with timelike boundary:

\begin{theorem}[\cite{Ake-Flores-Sanchez-18}, Theorem 1.1]\label{thm: Sanchez}
Any globally hyperbolic manifold with timelike boundary admits a Cauchy 
temporal function $t\colon \M\to\RR$ with gradient tangent to $\bM$.
This implies that $\M$ splits into $\RR \times \Sigma$ with metric 
$$ g= - \beta^2 d t^2\oplus h(t)\,,$$
where $\beta : \RR \times \Sigma \to \RR$ is a smooth positive function, $h(t)$ 
is a Riemannian metric on each slice
$\Sigma_t:=\{t\} \times \Sigma$ varying smoothly with $t$, and these slices are 
spacelike Cauchy hypersurfaces with boundary 
$\bSigma_t:=\{t\}\times\partial\Sigma$, namely achronal sets intersected 
exactly 
once by every inextensible timelike curve.
\end{theorem}

\subsection{Symmetric hyperbolic systems of constant characteristic}
Let now $\E\to\M$ be a Hermitian vector bundle over a globally hyperbolic manifold 
with timelike boundary $\M$, namely a $\KK$-vector bundle with finite rank $N$ 
endowed with a positive definite Riemannian or Hermitian fiber metric
$\fiber{\cdot}{\cdot}_p:\E_p\times\E_p\to\KK$.

\begin{definition}\label{def:symm syst}
 A linear differential operator $\oS \colon \Gamma(\E) \to \Gamma(\E)$ of first 
order is called a \textit{symmetric hyperbolic system} over $\M$ if 
\begin{enumerate}
\item[(S)]\label{conditionS} The principal symbol $\sigma_\oS (\xi) \colon \E_p \to \E_p$ is 
Hermitian with respect to $\fiber{\cdot}{\cdot}_p$ for every $\xi\in \T^*_p\M$ 
and 
for every $p \in \M$;
\item[(H)]\label{conditionH} For every future-directed timelike covector $\tau 
\in \T_p^*\M$, the 
bilinear form $\fiber{\sigma_\oS (\tau) \cdot}{\cdot}_p$ is positive definite on 
$\E_p$ for every $p\in\M$.
\end{enumerate}
Furthermore, we say that $\oS$ is \textit{of constant characteristic} if 
$\dim\ker \sigma_\oS(\n^\flat)$ is constant on the boundary.
In particular, if $ 
\sigma_\oS(\n^\flat)$ has maximal rank equal to $\mathrm{rk}(\E)=N$ everywhere 
on $\bM$ we say that $\oS$ is \textit{nowhere 
characteristic}.
\end{definition}

 \begin{remark}\label{rmk:H'}
  Notice that, if a system $\oS$ is hyperbolic with respect to a metric $g$ then 
it is also hyperbolic with respect to any metric in the conformal class of $g$. 
Indeed, conformal changes preserve each type of covector.
Furthermore, Condition (H)  implies that for any spacelike covector 
$\xi\in\T_p^*\M$  such that $\tau:=dt+\xi$ is timelike future-directed, 
\begin{align*}
\fiber{\sigma_\oS (dt) \cdot}{\cdot}_p + \fiber{\sigma_\oS (\xi) \cdot}{\cdot}_p   =\fiber{\sigma_\oS (dt+\xi) \cdot}{\cdot}_p >0
\end{align*}
 Therefore, a symmetric system is hyperbolic if and only if:
\begin{itemize}
\item[(H')] For every spacelike covector $ \xi 
\in \T_p^*\M$ such that $dt+\xi$ is a future-directed timelike covector, the 
bilinear form verifies 
\[\fiber{\sigma_\oS (\xi) \cdot}{\cdot}_p\quad >\quad - \fiber{ 
\sigma_\oS(dt)\cdot}{\cdot}_p.\]
\end{itemize}

 \end{remark}

\subsection{Admissible boundary conditions}

In order to discuss the Cauchy problem for a symmetric hyperbolic system, we 
have to impose suitable boundary conditions,
depending of course if we want to solve the forward or the backward Cauchy 
problem. 
We begin by fixing a Cauchy surface $\Sigma_0:=t^{-1}(\{0\})$ where we shall 
assign the initial data.
In this paper  we shall focus on a class introduced by 
Friedrichs and Lax-Phillips respectively in~\cite{Friedrichs,Lax-Phillips}, dubbed 
admissible 
boundary conditions. 

\begin{definition}\label{def:admissible bc}
A smooth linear bundle map $\pi_{\oB_+}\colon\E_{|_\bM 
}\rightarrow\E_{|_\bM}$ is 
said to be a {\it future admissible boundary condition} for a first-order 
Friedrichs system $\oS$ if 
\begin{itemize}
\item[(i-f)] the pointwise kernel $\oB_+$ of $\pi_{\oB_+}$ is a 
smooth subbundle of $\E_{|_\bM}$;
\item[(ii-f)] the quadratic form
$\Psi\mapsto\fiber{\sigma_{\oS}(\n^\flat)\Psi}{\Psi}_p$ is positive 
semi-definite 
on $\oB_+$ ;
\item[(iii-f)] the rank of $\oB_+$ is equal to the 
number of 
pointwise non-negative eigenvalues of 
$\sigma_{\oS}(\n^\flat)$ counting multiplicity.
\end{itemize}
 Similarly we say that $\pi_{\oB_-}\colon\E_{|_\bM}\rightarrow\E_{|_\bM}$ is 
{\it past admissible}
if
\begin{itemize}
\item[(i-p)] the pointwise kernel $\oB_-$ of $\pi_{\oB_-}$ is a 
smooth subbundle of $\E_{|_\bM}$;
\item[(ii-p)] the quadratic form
$\Psi\mapsto\fiber{\sigma_{\oS}(\n^\flat)\Psi}{\Psi}_p$ is negative 
semi-definite 
on $\oB_-$;
\item[(iii-p)] the rank of $\oB_-$ is equal to the 
number of 
pointwise non-positive eigenvalues of 
$\sigma_{\oS}(\n^\flat)$ counting multiplicity.
\end{itemize}
The pair $\oB=(\oB_+, \oB_-)$ is called the {\it admissible boundary space} or 
{\it admissible boundary condition}
for $\oS$.
\end{definition}

\begin{remark}
The role of $\oB_+$ and $\oB_-$ will become apparent when 
looking for energy estimates for symmetric hyperbolic $\oS$.
It turns out that $\oB_+$ (resp. $\oB_-$) is only needed in the future (resp. 
past) of the chosen Cauchy hypersurface $\Sigma_0$.
\end{remark}

Conditions (ii-f) and (ii-p) are equivalent to require that the boundary conditions are
{\it maximal} with respect to properties  (iii-f) and (iii-p) respectively, namely no smooth vector subbundles 
$(\oB')_\pm$ of $\E$ exist that properly contains $\oB_\pm$ and such that for 
all  
$\Phi'\in(\oB')_+$ and $\Phi''\in(\oB')_-$
$$\fiber{\sigma_{\oS}(\n^\flat)\Phi'}{\Phi'}\geq 0 \qquad 
\fiber{\sigma_{\oS}(\n^\flat)\Phi''}{\Phi''}\leq 0$$  
holds.  For further details we refer to~\cite[Section 2.2]{Ginoux-Murro-20}.

 With the next lemma, we shall see that admissible boundary conditions are 
``stable'' under conformal transformations, namely if $\oB$ is a future/past 
 admissible 
boundary space for a system on a globally hyperbolic manifold $(\M,g)$, then it 
is also future/past admissible for the same system on 
$(\M,\Omega^2 g)$, where $\Omega$ is a 
 positive smooth function on $\M$.
\begin{lemma}\label{lem:conf transf}
Let $(\M,g)$ be a globally hyperbolic spacetime with timelike boundary and let 
$\oB_\pm$ be a future/past admissible boundary space for a 
hyperbolic Friedrichs system of constant 
characteristic $\oS$. 
Then $\oB_\pm$ is future/past  admissible w.r.t $g$ if and only if 
it is 
 future/past admissible w.r.t.  $\Omega^2g$, for any positive 
$\Omega\in 
C^\infty(\M).$
\end{lemma} 
\begin{proof}
We only prove the case of a future admissible boundary condition, since the 
other case is analogous.
Let denote with $\n$ and $\tilde{\n}$ the normal vector w.r.t $g$ and $\Omega^2 
g$. Since $\tilde{\n}=\Omega^{-1}\n$, we get $\sigma_\oS(\n)=\Omega 
\sigma_\oS(\tilde{\n})$.  
This guarantees conditions(i-f)--(ii-f) in 
Definition~\ref{def:admissible bc} to be satisfied.
\end{proof}
Once a future/past admissible boundary condition $\pi_\oB$ is 
fixed, the \textit{adjoint boundary condition} $\pi_\oB^\dagger$ is defined as 
the 
pointwise orthogonal projection (with respect to $\fiber{\cdot}{\cdot}$) onto $\sigma_{\oS}
(\n^\flat)(\oB)$, namely
\begin{align}\label{Eq: adjoint bc}
\oB_+^\dagger
:=\left(\sigma_{\oS}(\n^\flat)(\oB_+)\right)^\perp\quad\qquad \oB_-^\dagger 
:=\left(\sigma_{\oS}(\n^\flat)(\oB_-)\right)^\perp\,.
\end{align}

\begin{definition}\label{def:selfadjointbc}
We say that an admissible boundary condition $\oB=(\oB_+,\oB_-)$ is 
\textit{self-adjoint} if and only if $\oB_+=\oB_-$.
\end{definition}

\begin{remark}
Our definition of self-adjoint boundary condition is actually stronger than the one used in the literature, where only $\oB_\pm=\oB^\dagger_\pm$ are required.
It immediately follows from the definition of a self-adjoint boundary condition 
that for any  $(\Psi,\Phi)\in\oB_+\oplus\oB_-$, it holds
$$\fiber{\sigma_{\oS}(\n^\flat)\Psi}{\Phi}= 0\,.$$
Actually, the vanishing of $(\Psi,\Phi)\mapsto\fiber{\sigma_{\oS}(\n^\flat)\Psi}{\Phi}$ on $\oB_+\oplus\oB_-$ is equivalent to $\oB_-\subset\oB_+^\dagger$ and hence $\oB_-=\oB_+^\dagger$ by identity of space dimensions.
As a consequence, if $\oB_+=\oB_-$, then $\oB_+^\dagger=\oB_+$ and $\oB_-^\dagger=\oB_-$.
Note however that both $\oB_+^\dagger=\oB_+$ and $\oB_-^\dagger=\oB_-$ do not imply $\oB_+=\oB_-$.
\end{remark}
\subsection{Well-posedness of the Cauchy problem}
Let
 $t\colon\M\to\RR$ be a Cauchy temporal function with 
gradient tangent to the boundary, as in Theorem~\ref{thm: Sanchez}, and write 
a symmetric system as
$$\oS=\sigma_{\oS}(dt)\nabla_{\partial_t} - \H \,$$
where $\H$ is a first-order linear differential operator which differentiates 
only in the directions that are tangent to $\Sigma$ and where $\nabla$ 
is 
 any fixed metric connection for
$\fiber{\cdot}{\cdot}$. 
Let  $\pi_{\oB_+},\pi_{\oB_-}\colon \E_{|_{\bM}}\longrightarrow\E_{|_{\bM}}$ be future and past 
admissible boundary conditions respectively for $\oS$, in 
particular their kernels define the 
future and past admissible boundary spaces $\oB_+,\oB_-$ respectively.
\begin{definition}\label{def:compcond}
{Let $\oS$ be a symmetric system $\oS$ with positive definite $\fiber{\sigma_\oS(dt)\cdot}{\cdot}$ and let $t_0\in\mathbb{R}$.
We say that $\h\in\Gamma(\E_{|_{\Sigma_{t_0}}})$ and $\f\in\Gamma(\E)$ fulfils 
the \textit{compatibility condition} of order $k\geq0$ at time 
$t_0\in\mathbb{R}$ if the following condition is satisfied:}
\begin{equation}\label{eq:comp cond data}
 \sum_{j=0}^k 
{k \choose j}
\Big(\nabla_{\partial_t}^j\pi_{\oB} \Big)\h_{k-j}\Big|_{\partial\Sigma_{t_0}} =0
\end{equation}
for both $\oB=\oB_+$ and $\oB=\oB_-$, where the sequence $(\h_k)_k$ of sections of $\E_{|_{\partial\Sigma_0}}$ is 
defined inductively by $\h_0:=\h$ and
$$ \h_k:= 
\sum_{j=0}^{k-1} 
{k-1\choose j}
 \H_j \, \h_{k-1-j}\Big|_{\partial\Sigma_{t_0}} + 
\nabla_{\partial_t}^{k-1} \big(\sigma_{\oS}^{-1}(dt)\f)\Big|_{\partial\Sigma_{t_0}}
\qquad 
\text{for all }k\geq1,
$$  
where $\H_j:=[\nabla_{\partial_t},\H_{j-1}]$ and $\H_0:=\sigma_{\oS}(dt)^{-1}\H$.
\end{definition}
Roughly speaking, Equation~\eqref{eq:comp cond data} provides a sufficient and necessary condition to ensure $C^k$-regularity for the solution of the Cauchy problem \eqref{Cauchy} once Cauchy data are given on $\Sigma_{t_0}$.
We recall one of the main results of \cite{Ginoux-Murro-20}, see \cite[Theorem 1.2]{Ginoux-Murro-20}:

\begin{theorem}[Smooth solutions for symmetric hyperbolic systems]\label{thm:SWHS}
Let $\M$ be a globally hyperbolic 
manifold with timelike boundary and let $\oS$ be a symmetric hyperbolic system of constant characteristic.
Assume $\oB=(\oB_+,\oB_-)$ to be an admissible boundary space for $\oS$ and let $\Sigma_{t_0}$ be any smooth spacelike Cauchy hypersurface in $\M$.
Then, for every $\f\in\Gamma_{c}(\E)$ and $h\in\Gamma_{c}(\E_{|_{\Sigma_{t_0}}})$ 
satisfying the compatibility conditions~\eqref{eq:comp cond data} up to any 
order,  there exists a unique $\Psi\in\Gamma(\E)$ satisfying the Cauchy problem
\begin{equation}\label{Cauchy}
\begin{cases}{}
{\oS }\Psi=\f  \\
\Psi_{|_{\Sigma_{t_0}}} = \h   \\
\Psi_{|_\bM\cap J^+(\Sigma_{t_0})}\in\oB_+\\
\Psi_{|_\bM\cap J^-(\Sigma_{t_0})}\in\oB_-\\
\end{cases} 
\end{equation}
 and the map $(\f,\h)\mapsto\Psi$ 
sending a pair $(\f,\h)\in\Gamma_{c}(\E)\times\Gamma_{c}(\E_{|_{\Sigma_{t_0}}})$ to the 
solution $\Psi\in\Gamma_{sc}(\E)$ of~\eqref{Cauchy} is continuous. 
 \end{theorem}

The assignment $U_{\oS,t}\colon D(U_{\oS,t})\subset\Gamma_{c}(\E|_{\Sigma_t})\ni\h\to U_{\oS,t}\h:=\Psi\in\Gamma_{sc}(\E)$ of a (unique) solution $\Psi$ to any smooth initial data $\h\in D(U_{\oS,t})$ is known as a \textit{Cauchy evolution operator} ---here $D(U_{\oS,t})$ is made by sections $\h\in\Gamma_{c}(\E|_{\Sigma_t})$ fulfilling the compatibility conditions \eqref{eq:comp cond data} with $\f=0$.
For later convenience we shall denote by $\rho_t\colon\Gamma(\E)\to\Gamma(\E|_{\Sigma_t})$ the restriction map for smooth sections: Notice that, $\rho_t$ is a right-inverse for $U_{\oS,t}$.
As shown in~\cite{cap,CapProp1,CapProp2,CapProp3}, on globally hyperbolic manifolds with empty boundary and compact Cauchy hypersurfaces, the evolution operator can be realized as a Fourier integral operator. As a matter of fact, the Fourier integral representation of the propagator contains the information on how singularities propagates in the manifold. As we shall see in Section~\ref{sec:Hadam}, this is of fundamental importance in proving the existence of Hadamard states for a free quantum field theory on a curved spacetime.

We conclude this section with the following result:
\begin{corollary}
Let $\M$ be a globally hyperbolic 
manifold with timelike boundary and let $\oB$ be an admissible boundary space for a 
symmetric hyperbolic system of constant characteristic $\oS$.
Then the Cauchy problem for $\oS$ on $(\M,g)$
is well-posed if and only if it is well-posed on 
$(\M,\Omega^2g)$ for any positive $\Omega\in C^\infty(\M)$.
\end{corollary}
\begin{proof}
Our claim follows immediately by Remark~\ref{rmk:H'} and Lemma~\ref{lem:conf transf}.
\end{proof}

\subsection{Symmetric weakly-hyperbolic systems}\label{sec:SWHS}
We conclude this section by showing that the Cauchy problem for a symmetric 
system $\oS$ is well-posed also if we assume that the principal symbol 
$\sigma_\oS(\xi)$ acts pointwise in a positive definite way only for a subset 
of future-directed timelike covectors $\xi$.
We begin with the following 
definition.
\begin{definition}\label{def:weakly-hyper}
A symmetric system of constant characteristic $\oS$ over $\M$ is 
\textit{weakly-hyperbolic} if there exists a positive smooth function 
$C \in C^\infty(\M)$
\begin{itemize}
\item[(gh)] The metric $g_C:=-\beta^2dt^2 \oplus C^2h(t)$ is globally hyperbolic, where $t$ is a Cauchy temporal function for $g$;
\item[(wH)] For any future-directed timelike covector 
$\tau$ of the form $\tau=dt+\xi \in \T_p^*\M$ it holds
$$\fiber{\sigma_\oS (dt+C\xi) \cdot}{\cdot}_p\quad>\quad0\,.$$
\end{itemize}
\end{definition}

\begin{remarks}\label{rem:wHandgalpha}
\noindent
\begin{enumerate}
\item The idea behind the Definition~\ref{def:weakly-hyper} is to `shrink' the light cone of the dual metric $g^\sharp$ in the cotangent bundle, so that the condition (H) in Definition~\ref{def:symm syst} has to be checked for a smaller class of future-directed timelike covectors (\cf Figure~\ref{fig:J alpha}).
Mind that, in the cotangent bundle, the causal future/past of $g_C$ is not allowed to shrink too much along any $\Sigma_t$.
\begin{figure}[h!]
\centering
\includegraphics[scale=0.35]{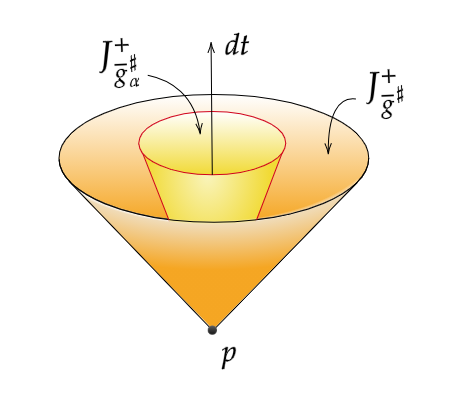}
\caption{The future light cones of $g^\sharp$ and  $g_C^\sharp$ in the cotangent bundle $\T^*\M$.}
\label{fig:J alpha}
\end{figure}
Clearly any symmetric hyperbolic system is a symmetric weakly-hyperbolic system, 
just take $C = 1$ in Definition \ref{def:weakly-hyper} (cf. 
Remark~\ref{rmk:H'}).  
\item 
Assuming the quadratic form $\fiber{\sigma_{\oS}(dt)\cdot}{\cdot}$ to be 
pointwise positive definite, there exists at each $p\in\M$ an 
$C(p)>0$ such that $\fiber{\sigma_{\oS}(\xi)\cdot}{\cdot}$ is 
positive definite for every $\xi\in J_{g_C}^+(p)$, where 
$J_{g_C}^+(p)\subset 
T_p^*\M$ is the causal future in $T_p^*\M$ w.r.t. the metric
$g_C:=-\beta^2dt^2\oplus C(p)^2h(t)$ defined on $T_p\M$, which is 
hence only defined 
at $p$.

\end{enumerate}
\end{remarks}

The following lemma shows that symmetric weakly-hyperbolic systems are not so far from symmetric hyperbolic systems.

\begin{lemma}\label{l:wHequivH}
Let $\oS$ be any symmetric weakly-hyperbolic system on a globally hyperbolic 
manifold $(\M,g)=(\mathbb{R}\times\Sigma,-\beta^2 
dt^2\oplus h(t))$ with or without any timelike boundary.
Let $C\in C^\infty(\mathbb{R},(0,\infty))$ be a function  depending only on time which satisfies
 {\rm (wH)} in Definition~\ref{def:weakly-hyper}.
Then $(\M,g_C):=(\mathbb{R}\times\Sigma,-\beta^2dt^2\oplus C^2 h(t))$ is globally hyperbolic and $\oS$ is symmetric hyperbolic on 
$(\M,g_C)$.
\end{lemma}

\begin{proof}
Let $p\in \M$ and let $\xi\in\T^*_p\M$ be $g_C$-timelike, that is, $g_C^\sharp(\xi,\xi)<0$ where $g_C^\sharp=-\beta^{-2}\partial_t^{\otimes 2}\oplus C^{-2}h(t)^\sharp$.
Then there exists unique $\lambda>0$ and $\check{\xi}\in T^*_{\pi_\Sigma(p)}\Sigma$ such that $\xi=\lambda\cdot(dt+C\check{\xi})$ ---here $\pi_\Sigma\colon\M=\mathbb{R}\times\Sigma\to\Sigma$ is the standard projection.
Condition $g_C^\sharp(\xi,\xi)<0$ is then equivalent to $g^\sharp(dt+\check{\xi},dt+\check{\xi})=g_C^\sharp(dt+C\check{\xi},dt+C\check{\xi})<0$, that is, to $dt+\check{\xi}$ being $g$-timelike.
Then condition (wH) implies that $\sigma_{\oS}(dt+C\check{\xi})=\lambda^{-1}\sigma_{\oS}(\xi)$ is positive definite.
This shows that $\oS$ is symmetric hyperbolic on $(\M,g_C)$.
Therefore, only the global hyperbolicity of $g_C$ on $\M$ remains to be proven.
For this, it suffices to show that $\beta^{-2}g_C=-dt^2\oplus\beta^{-2}C^2h(t)$ is globally hyperbolic when 
restricted to any subset of the form $(a,b)\times\Sigma$, with real $a<b$.
But since for all such $a,b$ there exists a positive constant $C_0$ such 
that $C(t)\geq C_0>0$ for all $t\in[a,b]$, we have $\beta^{-2}g_C\leq \beta^{-2}g_{C_0}$ on $[a,b]\times\Sigma$, where $g_{C_0}:=-\beta^2dt^2\oplus C_0^2h(t)$.
Therefore, it suffices to show that $\beta^{-2}g_{C_0}=-dt^2\oplus\beta^{-2}C_0^2h(t)$ is globally hyperbolic on 
$(a,b)\times\Sigma$.
But fixing any $t_0\in(a,b)$ and any inextensible $\beta^{-2}g_{C_0}$-timelike curve (which is 
$C^0$ and piecewise $C^1$) $\gamma=(\gamma_0,\hat{\gamma})$ in 
$(a,b)\times\Sigma$, the curve 
$\tilde{\gamma}:=(C_0^{-1}\gamma_0,\hat{\gamma})$ is $\beta^{-2}g$-timelike and still inextensible, therefore it meets $\{t_0\}\times\Sigma$ exactly once.
This shows $\beta^{-2}g_{C_0}$ and therefore $\beta^{-2}g_C$ (hence $g_{C}$) to be globally hyperbolic on $(a,b)\times\Sigma$.
This finishes the proof.
\end{proof}

\begin{example}
Let $\M=\mathbb{R}\times\Sigma$ be a globally hyperbolic manifolds and let $\underline{\mathbb{R}}$ be trivial line bundle over $\M$.
Then any future directed timelike vector field $X\in\Gamma(\T\M)$ defines an operator $\oS:=\nabla_X$ which is a symmetric weakly-hyperbolic system if and only if the projection $g(X,\partial_t)^{-1}X_{\Sigma_t}$ is bounded along $\Sigma_t$, where $X_{\Sigma_t}$ denotes the projection of $X$ on $\Sigma_t$ .
This applies in particular for $X=\partial_t+v$, $v\in\Gamma(\Sigma)$, the resulting transport equation being known with the name of Vlasov equation once applied in kinetic theory.
\end{example}

The definition \ref{def:admissible bc} of admissible boundary condition can be straightforwardly generalized for a symmetric weakly-hyperbolic system $\oS$.
The resulting connection with standard hyperbolic systems is described by the following lemma.

\begin{lemma}\label{lem:2}
Let $\oB$ be a future/past admissible boundary space for a  symmetric 
weakly-hyperbolic system $\oS$ over a globally hyperbolic manifold $(\M,g)$.
Then 
$\oB$ is future/past admissible for $\oS$ over $(\M,g_C)$. Furthermore, if $\oS$ 
is of constant characteristic on $(\M,g)$ then it is also of constant 
characteristic of $(\M,g_C)$.
\end{lemma}
\begin{proof}
To verify our claim is it enough to notice that the unit normal vectors to $\bM$ are the same up to a strictly positive smooth function.
This is due to the choice of a 
Cauchy temporal function whose gradient is tangent to the boundary.
\end{proof}

Combining Lemmas~\ref{l:wHequivH} and~\ref{lem:2}, we can conclude that the 
Cauchy problem for a symmetric weakly-hyperbolic system $\oS$ is well-posed. 
Indeed, these two lemmas guarantee that any smooth solution propagates with at 
most speed of light (w.r.t. $g_C$). 
Therefore, the Cauchy problem can be 
equivalently reformulated in terms of a Cauchy problem for a symmetric positive 
system with $\sigma_{\oS}(dt)>0$ in a globally hyperbolic manifold with compact 
Cauchy surfaces. 
We summarize our results in the following theorem and we leave 
the details to the reader.

\begin{theorem}[Smooth solutions for symmetric 
weakly-hyperbolic systems]\label{thm:main1}
Let $\M$ be a globally hyperbolic spacetime with timelike boundary and let 
$\oS$ be a symmetric weakly-hyperbolic system 
of constant characteristic.
Assume $\pi_{\oB_+},\pi_{\oB_-}$ to be future and past admissible boundary conditions for $\oS$.
Let $\Sigma_{t_0}$ be any smooth spacelike Cauchy hypersurface in $\M$.
Then, for every $\f\in\Gamma_{c}(\E)$ and $h\in\Gamma_{c}(\E_{|_{\Sigma_{t_0}}})$ 
satisfying the compatibility conditions~\eqref{eq:comp cond data} up to any 
order, there exists a unique $\Psi\in\Gamma(\E)$ satisfying the Cauchy problem
\begin{equation}\label{Cauchysmooth}
\begin{cases}{}
{\oS }\Psi=\f  \\
\Psi_{|_{\Sigma_{t_0}}} = \h   \\
\Psi_{|_{\bM\cap J^+(\Sigma_{t_0})}}\in\oB_+\\
\Psi_{|_{\bM\cap J^-(\Sigma_{t_0})}}\in\oB_-
\end{cases} 
\end{equation}
 and the map $(\f,\h)\mapsto\Psi$ 
sending a pair $(\f,\h)\in\Gamma_{c}(\E)\times\Gamma_{c}(\E_{|_{\bM}})$ to the 
solution $\Psi\in\Gamma_{sc}(\E)$ of \eqref{Cauchysmooth}, is continuous. 
 \end{theorem}

As usual, as a byproduct of the well-posedness of the Cauchy problem, we get 
the existence of Green operators.
 
\begin{proposition}\label{prop:Green}
A symmetric weakly-hyperbolic system $\oS$ of constant 
characteristic on a globally hyperbolic manifold with timelike boundary coupled 
with an admissible boundary 
condition $\oB=(\oB_+,\oB_-)$ is Green-hyperbolic, \ie, there exist linear maps, called 
advanced/retarded Green operator respectively,
$\G^\pm\colon \Gamma_{c}(\E)  \to \Gamma_{sc,\oB_\pm}(\E)$ satisfying
\begin{enumerate}[(i)]
\item\label{prop:Green1}
$\oS\circ \G^\pm \f=\f$ for all $ \f \in 
\Gamma_{c}(\E)$ and $\G^\pm\circ\oS\f=\f$ for all $ \f \in 
\Gamma_{c,\oB_\pm}(\E)$;
\item\label{prop:Green2}  $\supp(\G^\pm \f ) \subset J_{g_C}^\pm (\supp \f )$ 
for all $\f\in\Gamma_{c}(\E)\,,$
\end{enumerate}
where $J_{g_C}^\pm$ denote the causal future (+) and past ($-$) w.r.t. 
$g_C$ and $\Gamma_{\sharp,\oB_\pm}(\E)\subset\Gamma_{\sharp}(\E)$, 
$\sharp\in\{sc,c\}$ denotes the space of smooth  sections on $\E$ (with $\sharp$ 
support property) which fulfil the $\oB_\pm$-boundary condition.\\
Moreover, let $\mathsf{G}:=\mathsf{G}^+-\mathsf{G}^-\colon\Gamma_{c}(\E)\to\Gamma_{sc,\oB_++\oB_-}(\E)$ be the causal propagator associated with $\oS$ and $\oB$.
Then the following sequence is a complex
\begin{align*}
0\to
\Gamma_{c,\oB_+\cap\oB_-}(\E)
\stackrel{\oS}{\to}\Gamma_{c}(\E)
\stackrel{\mathsf{G}}{\to}\Gamma_{sc,\oB_++\oB_-}(\E)
\stackrel{\oS}{\to}\Gamma_{sc}(\E)\to 0
\end{align*}
which satisfies $\ker(\oS_{|_{\Gamma_{c,\oB_+\cap\oB_-}(\E)}})=\{0\}$, 
$\ker(\mathsf{G})=\oS\Gamma_{c,\oB_+\cap\oB_-}(\E)$ and 
$\oS\Gamma_{sc,\oB_++\oB_-}(\E)=\Gamma_{sc}(\E)$.
Moreover, if $\oB$ is self-adjoint, \ie $\oB_+=\oB_-$, then $\ker(\oS_{|_{\Gamma_{sc,\oB_+}(\E)}})=\mathsf{G}\Gamma_c(\E)$ and $\oS\colon\Gamma_{sc,\oB_+}(\E)\to\Gamma_{sc}(\E)$ is surjective, so that the complex is exact everywhere.
In that case, the solution space $\sol_{sc,\oB}(\oS):=\Gamma_{sc,\oB_+}(\E)\cap\ker(\oS)$ is characterized as
\begin{align}\label{Eq: characterization of solution space}
\sol_{sc,\oB}(\oS)
=\mathsf{G}\Gamma_{c}(\E)
\simeq\rquot{\Gamma_{c}(\E)}{\oS\Gamma_{c,\oB_+}(\E)}\,.
\end{align}
\end{proposition}

\begin{proof}
Properties \eqref{prop:Green1} and \eqref{prop:Green2} are satisfied by 
definition of $\mathsf{G}^\pm$.
As a straightforward consequence, $\mathsf{G}\oS=0$ on $\Gamma_{c,\oB_+\cap\oB_-}(\E)$ 
and $\oS\mathsf{G}=0$ on $\Gamma_{c}(\E)$, therefore the sequence is a 
complex.
Note that, since $(\mathsf{G}^\pm u)_{|_{\bM}}\in\oB_\pm$, we have $(\mathsf{G}u)_{|_{\bM}}\in\Gamma_{sc,\oB_++\oB_-}(\E)$ but beware that there is no reason why $\oB_++\oB_-=\E$ in general.\\
The injectivity of $\oS_{|_{\Gamma_{c,\oB_\pm}(\E)}}$ immediately follows from 
property \eqref{prop:Green1} since $\oS u=0$ for a $u\in \Gamma_{c,\oB_\pm}(\E)$ 
yields $u=\mathsf{G}^\pm\oS u=0$.
As a consequence, $\oS_{|_{\Gamma_{c,\oB_+\cap\oB_-}(\E)}}$ is injective.\\
To show that $\ker(\mathsf{G})\subset\oS\left(\Gamma_{c,\oB_+\cap\oB_-}(\E)\right)$, let 
$u\in \Gamma_{c}(\E)$ with $\mathsf{G}u=0$.
Then $\mathsf{G}^+u=-\mathsf{G}^-u$, so that $\supp\mathsf{G}^+u\subset 
J_{g_C}^+(\supp u)\cap J_{g_C}^-(\supp u)$ must be compact by property 
\eqref{prop:Green2}.
Moreover, because $(\mathsf{G}^\pm u)_{|_{\bM}}\in\oB_\pm$, we have $\mathsf{G}^+u\in\Gamma_{c,\oB_+\cap\oB_-}(\E)$.
Therefore $\mathsf{G}^+u\in\Gamma_{c,\oB_+\cap\oB_-}(\E)$ and 
satisfies $\oS\mathsf{G}^+u=u$ by property \eqref{prop:Green1}, from which 
$u\in\oS\left(\Gamma_{c,\oB_+\cap\oB_-}(\E)\right)$ follows.\\
From now on let us assume $\oB_+=\oB_-$ and prove that $\ker\left(\oS_{|_{\Gamma_{sc,\oB_+}(\E)}}\right)\subset\mathsf{G}\left(\Gamma_{c}
(\E)\right)$.
Let $u\in\Gamma_{sc,\oB_+}(\E)$ be such that $\oS u=0$.
By definition, there exists a compact subset $K$ of $\M$ such that $\supp 
u\subset J_{g_C}^+(K)\cup J_{g_C}^-(K)$.
Up to possibly enlarging $K$, we may assume that $\supp u\subset 
I_{g_C}^+(K)\cup I_{g_C}^-(K)$, where $I_{g_C}^+$ and $I_{g_C}^-$ denote 
the chronological future and past w.r.t. $g_C$ respectively.
Let $\{\chi_+,\chi_-\}$ be a partition of unity subordinated to the open 
covering $\{I_{g_C}^+(K), I_{g_C}^-(K)\}$ of $I_{g_C}^+(K)\cup 
I_{g_C}^-(K)$.
Let $u_\pm:=\chi_\pm u$.
Then $u=u_++u_-$, where each $u_\pm$ is smooth with $\supp u_\pm\subset 
I_{g_C}^\pm(K)$.
Furthermore, because $u_\pm$ is obtained by pointwise multiplication of $u$ by  
a real number, we have $u_\pm{}_{|_{\bM}}\in\oB_+$.
Let $v:=\oS u_+(=-\oS u_-)$.
Then $v$ is smooth with support contained in $J_{g_C}^+(K)\cap J_{g_C}^-(K)$, 
therefore $\supp v$ is compact.
We would like to check that $\mathsf{G}v=u$ in the weak -- and therefore also in the strong 
-- sense.
For that, we need the following fact: if 
$(\mathsf{G}^\pm)^*$ denotes the 
formal adjoint of $\mathsf{G}^\pm$, then actually
\[(\mathsf{G}^\pm)^*=\mathsf{G}_\dagger^\mp\]
holds, where $\mathsf{G}_\dagger^+$ and $\mathsf{G}_\dagger^-$ are the Green operators for $\oS^\dagger$ with boundary condition $\oB^\dagger:=(\oB_+^\dagger,\oB_-^\dagger)$.
Recall that, if $\oB$ is a future/past admissible boundary condition for $\oS$, then  
$\oB^\dagger$ is a future/past admissible boundary condition for $\oS^\dagger$.
Moreover, $\oS^\dagger$ becomes a symmetric weakly hyperbolic system on $\M$ 
with reversed time 
orientation, in particular $\oS^\dagger$ has unique advanced and retarded Green 
operators as well.\\
To check that $(\mathsf{G}^\pm)^*=\mathsf{G}_\dagger^\mp$, let $\varphi,\psi$ 
be arbitrary in $\Gamma_{c}(\E)$.
Because of 
$\supp\mathsf{G}^\pm\varphi\cap\supp\mathsf{G}_\dagger^\mp\psi$ being compact, 
we may perform the following partial integration on $\M$:
\begin{eqnarray*}
\int_{\M}\fiber{\mathsf{G}^\pm\varphi}{\psi}\vol_{\M}&=&\int_{\M}\fiber{\mathsf{
G}^\pm\varphi}{\oS^\dagger\mathsf{G}_\dagger^\mp\psi}\vol_{\M}\\
&=&\int_{\M}\fiber{\oS\mathsf{G}^\pm\varphi}{\mathsf{G}_\dagger^\mp\psi}\vol_{\M
}-\int_{\bM}\fiber{\sigma_{\oS}(\n^\flat)\mathsf{G}^\pm\varphi}{\mathsf{G}
_\dagger^\mp\psi}\vol_{\bM}\\
&=&\int_{\M}\fiber{\varphi}{\mathsf{G}_\dagger^\mp\psi}\vol_{\M
},
\end{eqnarray*}
where the boundary term vanishes by $\mathsf{G}^\pm\varphi_{|_{\bM}}\in\oB_\pm$ and 
$\mathsf{G}_\dagger^\mp\psi_{|_{\bM}}\in\oB_\pm^\dagger$.
This shows $(\mathsf{G}^\pm)^*=\mathsf{G}_\dagger^\mp$.
Now given any $\psi\in\Gamma_{c}(\E)$, we have 
\begin{eqnarray*}
\int_{\M}\fiber{\mathsf{G}^\pm 
v}{\psi}\vol_{\M}&=&\int_{\M}\fiber{v}{(\mathsf{G}^\pm)^* \psi}\vol_{\M}\\
&=&\int_{\M}\fiber{v}{\mathsf{G}_\dagger^\mp \psi}\vol_{\M}\\
&=&\pm\int_{\M}\fiber{\oS u_\pm}{\mathsf{G}_\dagger^\mp \psi}\vol_{\M}\\
&=&\pm\int_{\M}\fiber{u_\pm}{\oS^\dagger\mathsf{G}_\dagger^\mp \psi}\vol_{\M}\qquad\textrm{because }u_\pm{}_{|_{\bM}}\in\oB_+=\oB_+^\dagger\\
&=&\pm\int_{\M}\fiber{u_\pm}{\psi}\vol_{\M},
\end{eqnarray*}
where we have used in a crucial way that $\mathsf{G}_\dagger^\mp 
\psi_{|_{\bM}}\in\oB_\mp^\dagger$ and that $u_\pm{}_{|_{\bM}}\in\oB_+$ as well as $\oB_+^\dagger=\oB_+$.
Therefore, $\mathsf{G}^\pm v=\pm u_\pm$ and
$\mathsf{G}v=u_++u_-=u$, as we claimed.\\
It remains to look at the surjectivity of 
$\oS\colon\Gamma_{sc,\oB_+}(\E)\to\Gamma_{sc}(\E)$, still with the assumption that $\oB_+=\oB_-$.
Let $\f\in\Gamma_{sc}(\E)$ be given and $K\subset\M$ be compact such that 
$\supp\f\subset J_{g_C}^+(K)\cup J_{g_C}^-(K)$.
As above, up to enlarging $K$ we may assume that $\f=\f_++\f_-$, where 
$\f_\pm\in\Gamma_{sc}(\E)$ with $\supp\f_\pm\subset J_{g_C}^\pm(K)$.
By Theorem \ref{thm: Sanchez} the spacetime $\M$ can be identified with $\RR\times\Sigma$, where $\Sigma$ a smooth spacelike Cauchy hypersurface of $\M$.
For each $n\in\mathbb{N}$ we
let $\M_{(-n,n)}:=(-n,n)\times\Sigma$, where 
$\Sigma\simeq\{0\}\times\Sigma$.
Note that $\M_{(-n,n)}$ is still a globally hyperbolic spacetime with timelike boundary.
Let $\chi_n$ be a smooth function with timelike compact support such that 
$\chi_n{}_{|_{\M_{(-n,n)}}}=1$.
Then $\chi_n\f_+$ lies in $\Gamma_c(\E)$ and we may consider $u_n:=\mathsf{G}^+\chi_n \f_+\in\Gamma_{sc,\oB_+}(\E)$.
Now letting $u^+(x):=u_n^+(x)$ for every $x\in\M_{(-n,n)}$ defines a smooth 
section of $\E$ on $\M$ with $\oS u^+=\f_+$, for if e.g. $m>n$ then
$v:=u_m^+-u_n^+$ is a smooth spacelike compact section of $\E$ satisfying $\oS 
v=0$ on $\M_{(-n,n)}$ as well as $v_{|_{\M\setminus J^+(\supp\f_+)}}=0$ and 
$v_{|_{\bM_{(-n,n)}}}\in\oB_+$, so that $v=0$ on $\M_{(-n,n)}$ by uniqueness of 
the solution of the forward Cauchy problem.
The support of $u^+$ must be contained in 
$J_{g_C}^+(K)$ since this is the case for the support of each $u_n^+$.
Analogously, there exists a $u^-\in\Gamma_{sc,\oB_-}(\E)=\Gamma_{sc,\oB_+}(\E)$ with $\oS u^-=\f^-$ and 
therefore $\oS(u^++u^-)=\f$.
This proves the surjectivity of 
$\oS\colon\Gamma_{sc,\oB_+}(\E)\to\Gamma_{sc}(\E)$ and concludes the proof of 
Proposition \ref{prop:Green}.
\end{proof}

For further details we refer to~\cite[Proposition 5.1]{Ginoux-Murro-20}, \cite[Proposition 20]{Dappiaggi-Drago-Longhi-20} and
\cite[Proposition 36]{Dappiaggi-Drago-Ferreira-19}.

\subsection{M\o ller operators on manifolds with timelike boundary}\label{sec:Moller}

In~\cite{defarg} a geometric process was realized to compare solutions of 
symmetric hyperbolic systems
on different globally hyperbolic manifolds $\M_0:=(\M,g_0)$ and $\M:=(\M,g_1)$ 
with empty boundary, provided that $\M_0$ and $\M_1$ admit the same Cauchy 
temporal function and $g_1\leq g_0$, namely the set of timelike vectors for 
$g_1$ is contained in the one for $g_0$. 
This was achieved via the construction of a family of so-called M\o ller operators \cite{Moller,Drago-Hack-Pinamonti-17,HW-05}.
The aim of this paper is to generalize that construction on manifolds with 
timelike boundary, where the assumption $g_1\leq g_0$ is adapted 
to the situation.
\medskip

\noindent Let us introduce the following setup:

\begin{setup}\label{setup}
\noindent\begin{enumerate}[(i)]
\item 
{$\M_0=(\M,g_0)$ and $\M_1=(\M,g_1)$ are globally hyperbolic manifolds 
with timelike boundary and with the same Cauchy temporal function 
$t\colon\M\to\RR$.
Moreover, by realizing $(\M,g_i)=(\mathbb{R}\times\Sigma,-\beta_i^2dt^2\oplus 
h_i(t))$ for $i=0,1$ ---\textit{cf.} Theorem \ref{thm: Sanchez}--- we shall 
assume that there exists a smooth positive function $C> 0$ such that} 
\[C^2 \beta_1^{-2} {h_1}(t) \leq \beta_0^{-2} 
{h_0}(t)\]
holds for every $p\in\M$ and $g_C:=-\beta_1^2 dt^2 \oplus C^2 h_1(t)$ is globally hyperbolic;
\item $\E_1$ (resp. $\E_0$) is a $\KK$-vector bundle over $\M_1$ (resp. $\M_0$)  
with finite rank and endowed with a nondegenerate bilinear or sesquilinear 
fiber metric $\fiber{\cdot}{\cdot}_1$ (resp. $\fiber{\cdot}{\cdot}_0$);
\item 
{$\kappa_{1,0}\colon \E_0\to \E_{1}$ is a fiberwise linear isometry of vector bundles with inverse $\kappa_{0,1}\colon\E_1\to\E_0$.
With a slight abuse of notation we shall also denote by $\kappa_{1,0}\colon \Gamma(\E_0)\to \Gamma(\E_{1})$ the corresponding linear map between sections defined by $[\kappa_{1,0}u](x)=\kappa_{1,0}u(x)$ for all $u\in\Gamma(\E_0)$ and $x\in\M$.
Finally for any positive $f\in C^\infty(\M)$, we set 
$\bkappa_{1,0}:=f\; \kappa_{1,0}\colon\Gamma(\E_0)\to\Gamma(\E_1)$ with inverse $\bkappa_{0,1}:=f^{-1}\; \kappa_{0,1}$;}

\item\label{assumption:S0andS1} $\oS_1$ (\textit{resp.} $\oS_0$) is a symmetric 
weakly-hyperbolic system with self-adjoint 
admissible boundary space we shall denote by $\oB_1$ (\textit{resp.} $\oB_0$).
Moreover we shall assume that $\dim\ker\sigma_{\oS_1}(\xi)$ is 
constant for any nonzero spacelike covector $\xi\in\T^*\M_1$;
\item\label{assumption:S01}  Let $\oS_{0,1}^f\colon\Gamma(\E_1)\to\Gamma(\E_1)$ be 
the operator defined by 
$\oS_{0,1}^f:=\bkappa_{1,0}\oS_0\bkappa_{0,1}$.
We assume there exists a linear isometry 
 $\wp_{1,0}\colon\T^*\M_0\to\T^*\M_1$ 
such that 
$ \sigma_{\oS_{0,1}^f}(\xi) = \sigma_{\oS_1}(\wp_{1,0}\xi)$ for every 
$\xi\in\T^*\M_1$.
\end{enumerate}
\end{setup}

\begin{remarks}\label{rmk:JC}
\noindent\begin{enumerate}
\item          
The assumption (i) in the Setup~\ref{setup} can be equivalently rephrased  as 
follows. 
Consider the metric $g_C:= -\beta_1^{2}dt^2 \oplus C^2  h_1(t)$ 
which is globally hyperbolic on account of Definition~\ref{def:weakly-hyper}. Then we get 
the following two situations: for any vector $v\in\T\M$, we get
\begin{itemize}
\item[($ C\leq 1$)]   $ g_C(v,v) \leq g_0(v,v)$ and $ g_C(v,v) \leq g_1(v,v)$, which implies that $J^\pm_{g_0}\cup J^\pm_{g_1}\subset J^\pm_{g_C}$.
\item[($C\geq 1$)] 
 $g_1(v,v)\leq g_C(v,v) \leq  g_0(v,v) $, which implies that $J^\pm_{g_0}\subset J^\pm_{g_C} \subseteq J^\pm_{g_1}$.
\begin{figure}[h!]
\vspace{-8mm}
\begin{subfigure}[t]{8.5cm}
\centering
\includegraphics[scale=0.35]{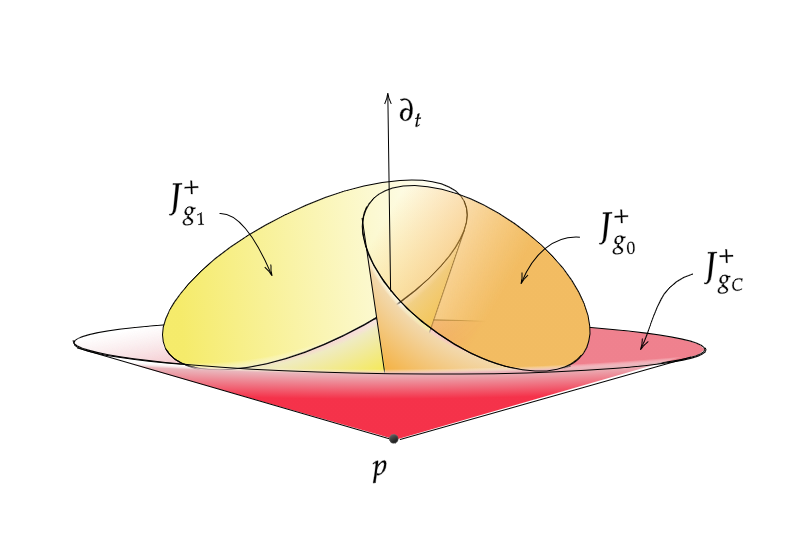}
\caption{$0<C\leq 1$}\label{fig:2a}
\end{subfigure}
\begin{subfigure}[t]{9cm}
\centering
\includegraphics[scale=0.34]{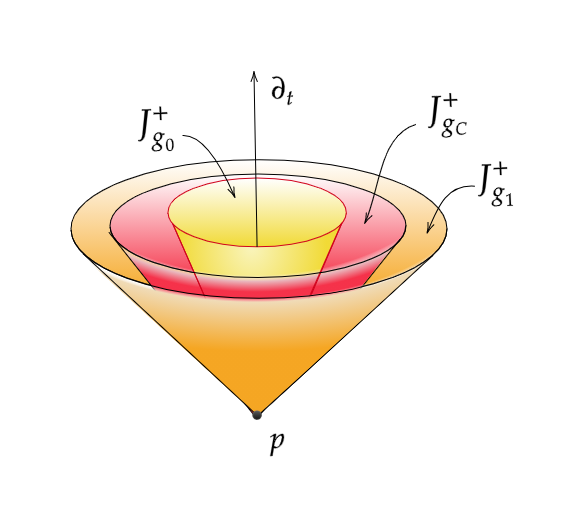}
\caption{$C>1$}\label{fig:2b}
\end{subfigure}
\caption{Future light cones of $g_0$ and $g_1$ satisfying $\beta^{-2}_1h_1(t)\leq C^2 \beta_0^{-2} h_0(t)$.}
\end{figure}
\end{itemize}
\item Using assumption \eqref{assumption:S0andS1}, assumption 
\eqref{assumption:S01} implies both that $\dim\ker\sigma_{\oS_0}(\xi)$ is 
constant for any nonzero spacelike covector $\xi$ and that $\wp_{1,0}$ is 
time-orientation preserving.
\end{enumerate}
\end{remarks}

\begin{remark}[The principal symbol of $\oS_{0,1}:=\oS_{0,1}^1$]\label{rem:symb}
For later convenience we shall compute the principal symbol $\sigma_{\oS_{0,1}}$ of $\oS_{0,1}$ in a slighly more general framework that the one depicted above.
(Or maybe this is not at all necessary and we set $\zeta=\operatorname{id}_\M$ from the beginning.)
Let $\E_0\to \M_0$ and $\E_1\to \M_1$ be two vector bundles such that there 
exists a diffeomophism $\zeta_{1,0}\colon \M_0\to \M_1$ with (inverse $\zeta_{0,1}$) which is lifted to a vector bundle isomorphism
$\kappa_{1,0}\colon \E_0\to \E_1$.
With a slight abuse of notation we shall denote with
$\kappa_{1,0}\colon\Gamma(\E_0)\to\Gamma(\E_1)$
the associated map of vector bundles
defined by
\begin{align*}
(\kappa_{1,0} u_0)(x_1):=\kappa_{1,0} (u_0(\zeta_{0,1}x_1))\,.
\end{align*}
for all $u_0\in\Gamma(\E_0)$ and $x_1\in \M_1$.
Notice that  $\kappa_{1,0}(fu_0)=(\zeta_{0,1}^*f)\kappa_{1,0}u_0$ for all $u_0\in\Gamma(\E_0)$ and $f\in C^\infty(\M_0)$ where $\zeta_{0,1}^*\colon C^\infty(\M_0)\to C^\infty(\M_1)$.
Moreover, $\kappa_{1,0}\colon\Gamma(\E_0)\to\Gamma(\E_1)$ is invertible with inverse $\kappa_{0,1}$.

The principal symbol of $\oS_{0,1}:=\kappa_{1,0}\oS_0\kappa_{0,1}$ is obtained as follow. For all $u_1\in \E_1|_{x_1}$ and $\xi_1\in \T^*_{x_1}\M_1$, let $\tilde{u}_1\in\Gamma(\E_1)$ and $f\in C^\infty(\M_1)$ be such that $\tilde{u}_1(x_1)=u_1$ and $\mathrm{d}f(x_1)=\xi_1$.
Then we have
\begin{align*}
\sigma_{\oS_{0,1}}(\xi_1)u_1=
[\kappa_{1,0}\oS_0\kappa_{0,1},f_1]\tilde{u}_1|_{x_1}
&=\kappa_{1,0}\oS_0\kappa_{0,1} f_1\tilde{u}_1
-f_1\kappa_{1,0}\oS_0\kappa_{0,1} \tilde{u}_1
\\&=\kappa_{1,0} \oS_0 (\zeta_{1,0}^*f_1\kappa_{0,1}\tilde{u}_1)
-f_1\kappa_{1,0}\oS_0\kappa_{0,1}\tilde{u}_1
\\&=\kappa_{1,0}[\oS_0,\zeta_{1,0}^*f_1]\kappa_{0,1}\tilde{u}_1
\\&=\kappa_{0,1}\sigma_{\oS_0}(\mathrm{d} (\zeta_{1,0}^*f_1))\kappa_{0,1}\tilde{u}_1|_{x_1}
\\&=\kappa_{0,1}\sigma_{\oS_0}((\mathrm{d}\zeta_{1,0}^*\mathrm{d}f_1)\kappa_{0,1}\tilde{u}_1|_{x_1}
\\&=\kappa_{0,1}\sigma_{\oS_0}((\mathrm{d}\zeta_{1,0}^*\xi_1)\kappa_{0,1}\tilde{u}_1|_{x_1}\,,
\end{align*}
where $(\mathrm{d}\zeta_{1,0})^*\colon \T^*\M_1\to \T^*\M_0$.
Overall we have
\begin{align*}
\sigma_{\oS_{0,1}}(\xi_1)
=\kappa_{1,0}\sigma_{\oS_0}[(\mathrm{d}\zeta_{1,0})^*\xi_1]\kappa_{0,1}\,.
\end{align*}
\end{remark}

Similarly to the case of an empty boundary, the construction of a family of M\o ller operators 
{requires to control the Cauchy problem for the operator $\oS_{0,1}^f$}. 
With the next proposition, we shall show that the $\oS_{0,1}^f$ is actually symmetric weakly-hyperbolic over $\M_1$.

\begin{proposition}\label{prop: S01 is wSHS}
Assume the Setup~\ref{setup}. 
Then the operator $\oS_{0,1}^f$ is a symmetric weakly-hyperbolic system of 
constant characteristic on $\M_1$ and 
$\bkappa_{1,0}(\oB_0)=\kappa_{1,0}(\oB_0)$ is a self-adjoint   admissible boundary space for 
$\oS_{0,1}^f$.
\end{proposition}
\begin{proof}
On account of Remark \ref{rem:symb} --- with 
$\zeta_{0,1}=\operatorname{id}_\M$ --- we have that, for every $\xi\in T^*\M$,
$\sigma_{\oS_{0,1}^f}(\xi)=\bkappa_{1,0}\sigma_{\oS_0}
(\xi)\bkappa_{0 , 1}=\kappa_{1,0}\sigma_{\oS_0}
(\xi)\kappa_{0 , 1}$.
Since $\oS_0$ is symmetric and $\kappa_{1,0}$ is a fiberwise linear isometry by 
assumption, $\oS_{0,1}^f$ clearly satisfies property (S) in 
Definition~\ref{def:symm syst}.
Moreover, because $\oS_0$ has constant characteristic and $\n_1^\flat$ is a 
pointwise positive scalar multiple of $\n_0^\flat$, the operator $\oS_{0,1}^f$ 
has constant characteristic.
Because $\oB_0$ is an admissible boundary condition for $\oS_0$, the subbundle 
$\kappa_{1,0}(\oB_0)=\bkappa_{1,0}(\oB_0)$ must be an admissible boundary space 
for $\oS_{0,1}^f$,  and it remains self-adjoint.
We shall next verify property (wH) in Definition~\ref{def:weakly-hyper}. 
To this  end let 
$g_{0,C_0}=-\beta_0^2dt^2\oplus C_0^2h_0(t)$ be the globally hyperbolic metric chosen 
accordingly with 
Lemma~\ref{l:wHequivH}. Then $\oS_0$ is a symmetric hyperbolic system and 
$\sigma_{\oS_0}(\tau)$ is fiberwise positive definite for any future-directed 
(w.r.t. $g_{0,C_0}$) covector $\tau$. 
Since any conformal transformation does not 
change the set of future-directed covectors, the operator $\oS_0$ is hyperbolic 
w.r.t. $\overline{g}_0:=\beta_0^{-2} 
g_{0,C_0} = -dt^2 \oplus 
C_0^2 \beta_0^{-2}{h_0}(t)$. 
We shall now prove that $\oS_{0,1}^f$ is symmetric hyperbolic with respect to the metric $\overline{g}_1:= -dt^2 \oplus C_0^2 C(t)^{-2}\beta^{-2}_1{h_1}(t)$.
For that, let $\tau=dt+\xi$ be $\overline{g}_1$-timelike future directed: On account of the assumption $\beta_1^{-2}{h_1}(t) \leq C^2(t) 
\beta_0^{-2} {h_0}(t)$ we find
\begin{align*}
\overline{g}_0^\sharp(dt+\xi,dt+\xi)
\leq\overline{g}_1^\sharp(dt+\xi,dt+\xi)
<0\,,
\end{align*}
so that $dt+\xi$ is $\overline{g}_0$-timelike future directed ---notice that $\overline{g}_0^\sharp=\partial_t^{\otimes 2}\oplus C_0^{-2}C(t)^2\beta_0^2h_0(t)^\sharp$ and similarly for $\overline{g}_1$ (\cf Figure~\ref{fig:Jsharp}).
\begin{figure}[h!]
\centering
\includegraphics[scale=0.35]{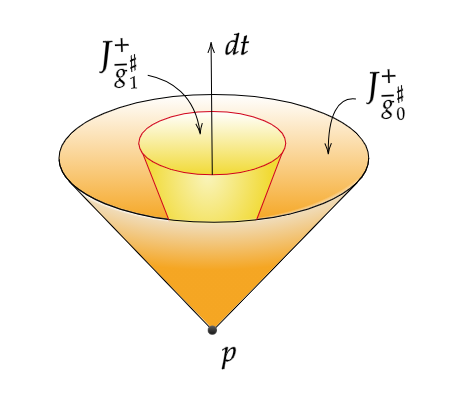}
\caption{The future light cones of $g_1^\sharp$ and  $g_0^\sharp$ in the cotangent bundle $\T^*\M$.}
\label{fig:Jsharp}
\end{figure}

\noindent It follows that $\sigma_{\oS_0}(dt+\xi)>0$ and therefore 
$\sigma_{\oS_{0,1}^f}(dt+\xi)>0$ as well. 
This shows that $\oS_{0,1}^f$ is symmetric hyperbolic with respect to $\overline{g}_1$ and therefore the same holds true for $g_{1,C_1}:=\beta_1^2\overline{g}_1=-\beta_1^2dt^2\oplus C_1^2 h_1(t)$, where $C_1^2:=C_0^2 C(t)^{-2}>0$ on account of the hypothesis on $C$ and $C_0$.
This proves that $\oS_{0,1}^f$ is weakly-hyperbolic with respect to $g_1$ as per Definition \ref{def:weakly-hyper}.
\end{proof}

Note that the existence of a linear isometry $\wp_{1,0}$ from assumption 
(\ref{assumption:S01}) is not required in the proof of Proposition 
\ref{prop: S01 is wSHS}.\\

So far, we considered a setting where the operators $\oS_0,\oS_1$, though being 
defined on different bundles, can be compared through $\kappa_{1,0}$.
As a matter of fact the next step in the construction of a M\o ller operator 
intertwining $\oS_0$ and $\oS_1$ is to build an intertwining operator between 
$\oS_{0,1}^f$ and $\oS_1$.
To this avail, we shall first consider an interpolating operator $\oS_\chi^f$ 
defined by 
$\oS_{\chi,1}^f:=(1- \chi)\oS_{0,1}^f+\chi\oS_1$,
where $\chi\in C^\infty(\M)$ is an arbitrary smooth function with 
$0\leq\chi\leq 1$.

{The following proposition ensures that} $\oS_{\chi,1}^f$ is a symmetric 
weakly-hyperbolic system of constant characteristic as soon as 
$\wp_{1,0}(\n_1^\flat)$ is not any pointwise negative scalar multiple of 
$\n_1^\flat$.

\begin{proposition}\label{prop:1}
Assume the Setup~\ref{setup} and that $\wp_{1,0}(\n_1^\flat)\neq \mu \n^\flat_1$ for any $\mu <0$. Then for any  $\chi\in C^\infty(\M,[0,1])$, the operator defined by 
\begin{equation}\label{eq:def S intert}
\oS_{\chi,1}^f := (1- \chi) \,\oS_{0,1}^f + \chi\oS_{1}+ \frac{1}{2}\left(\sigma_{\oS_1}-\sigma_{\oS_{0,1}^f}\right)(d\chi) \colon\Gamma(\E_1)  \to 
\Gamma(\E_1) 
\end{equation}
is a symmetric weakly-hyperbolic system of constant characteristic over $\M_1$. 
\end{proposition}
\begin{proof}
By Remark \ref{rem:symb} it follows that
\begin{align*}
\sigma_{\oS_{\chi,1}^f}(\xi)
=(1-\chi)\sigma_{\oS_{0,1}^f}(\xi)
+\chi\sigma_{\oS_{1}}(\xi)\,.
\end{align*}
Therefore, $\oS_{\chi,1}^f$ is a symmetric system.
We shall next  notice that a convex combination of weakly-hyperbolic system is still weakly-hyperbolic.
{As a matter of fact if $C_{0,1},C_1\in(0,1]$ denotes the positive 
 functions of Definition \ref{def:weakly-hyper} associated with $\oS_{0,1}^f$ 
and $\oS_1$, it follows that}
$$\fiber{\sigma_{\oS_{\chi,1}^f} (dt+C\xi) \cdot}{\cdot}_p >0\,,$$
where $C=\min\{C_{0,1},C_1\}$ for every future-directed 
$g_1$-timelike covector $\tau=\mathrm{d}t+\xi$.

To conclude our proof, we shall show that $\oS_{0,1}^f$ is of constant 
characteristic.
To this end, we consider
\begin{align*}
 \sigma_{\oS_{\chi,1}^f}(\n_1^\flat)= (1-\chi)\sigma_{\oS_{0,1}}(\n_1^\flat) + \chi \sigma_{\oS_1}(\n_1^\flat) = \sigma_{\oS_1}((1-\chi)\wp_{1,0}\n_1^\flat + \chi\n_1^\flat )\,.
\end{align*}
Since we have by assumption that $\wp_{1,0}\n_1^\flat\neq \mu \n_1^\flat$ for 
any $\mu<0$, the covector $(1-\chi)\wp_{1,0}\n_1^\flat + \chi\n_1^\flat$ is non-vanishing along $\bM$, which implies that $(1-\chi)\wp_{1,0}\n_1^\flat + 
\chi\n_1^\flat $ is a nonzero spacelike covector.
In particular, by assumption 
(\ref{assumption:S0andS1}) in Setup~\ref{setup}, $\oS_{\chi,1}^f$ is of constant 
characteristic.
\end{proof}

\begin{remarks}
\noindent
\begin{enumerate}
\item Note that the zero-order operator 
$V:=\frac{1}{2}\left(\sigma_{\oS_1}-\sigma_{\oS_{0,1}^f}\right)(d\chi)$ is a Hermitian operator 
which vanishes on every open subset where $\chi$ is constant -- in particular on 
both the chronological past of $\Sigma_-$ and the chronological future of 
$\Sigma_+$.
The zero-order operator $V$ does not play any role in the proof of Theorem \ref{thm:Moller}.
However, whenever $\oS_1$, $\oS_0$ are formally skew-adjoint, the presence of $V$ ensures that $\oS_{\chi,1}^f$ is formally skew-adjoint provided a suitable choice of $f$ is made ---\textit{cf.} Proposition \ref{prop:Schi skewad} for the precise statement.
\item The Assumption (\ref{assumption:S0andS1}) in Setup~\ref{setup} is needed in order to ensure that $\oS_{\chi,1}^f$ is of constant characteristic.
It can be dropped if $\wp_{1,0}\n^\flat_1=\n^\flat_1=\n^\flat_0$.
\end{enumerate}
\end{remarks}

Building on Proposition~\ref{prop:1}, we now prove the main result of this paper.
\begin{theorem}\label{thm:Moller}
Assume the Setup~\ref{setup} and that $\wp_{1,0}(\n_1^\flat)\neq \mu \n^\flat_1$ 
for any $\mu <0$. 
Consider two Cauchy hypersurfaces $\Sigma^\pm\subset \M_1$ such that
$\Sigma_+ \subset J_{g_1}^+(\Sigma_-)$ --- where $J_{g_1}^\pm$ denote the causal 
cones w.r.t. $g_1$ --- and let $\chi \in C^\infty(\M_1,[0,1])$ be 
non-decreasing along any future-oriented timelike curve such that 
$$\chi_{|_{J_{g_1}^+(\Sigma_+)}}=1 \,,\qquad \text{and } \qquad 
\chi_{|_{J_{g_1}^-(\Sigma_-)}}=0 \,.$$
Finally let $\oB_\chi$ be a self-adjoint admissible boundary space for $\oS_{\chi,1}^f$ such 
that
$$\oB_\chi=\begin{cases}
\oB_{0,1}:=\bkappa_{1,0}(\oB_0) & \text{ where $\chi=0$}\\
\oB_1 & \text{ where $\chi=1$}\\
\end{cases}\,.
$$
Then the Cauchy problem for $\oS_{\chi,1}^f$ with $\oB_\chi$-boundary conditions is well-posed.
Moreover, let $U_{\oS_{\chi,1}^f,\pm}\colon D(U_{\oS_{\chi,1}^f,\pm})\subset\Gamma_{c}(\E_1|_{\Sigma_\pm})\to\Gamma_{sc}(\E_1)$ be the Cauchy evolution operator associated with $\oS_{\chi,1}^f$ and initial data on $\Sigma_\pm$ and let $\rho_\pm\colon\Gamma_{sc}(\E_1)\to\Gamma_{c}(\E_1|_{\Sigma_\pm})$ be the standard restriction maps.

Then the M\o ller operator 
{$\R_{0,1}= U_{\oS_1,+} \circ \rho_+ \circ U_{\oS_{\chi,1}^f,-} \circ \rho_- \circ\kappa^f_{1,0}\,,$}
implements an 
isomorphism between the spaces of solutions 
 of $\oS_0$ 
and $\oS_1$
 given by 
$$ \sol_{\oB_C}(\oS_C):=\{\Psi_C\in \Gamma(\E_C) \,| \,   \oS_C\Psi_C = 0 \;\text{ and }\; \Psi_C|_\bM\in \oB_C\;  \} \qquad \text{for $C=0,1$}\,.$$
\end{theorem}
\begin{proof}
Since $\oB_\chi$ is a self-adjoint admissible boundary space, the Cauchy evolution operators and the Cauchy data map are well-defined on account of Theorem~\ref{thm:main1}.
Furthermore, for any $\Psi_0\in\sol_{\oB_0}(\oS_0)$ we have $\rho_-\bkappa_{1,0}\Psi_0\in D(U_{\oS_{\chi,1}^f,-})$ because of $\oB_\chi$ coincide with $\bkappa_{1,0}(\oB_0)$ on $\Sigma_-$.
Therefore $U_{\oS_{\chi,1}^f,-}\rho_-\kappa^f_{1,0}\Psi_0$ is well defined.
For a similar reason, for any $\Psi\in\sol(\oS_{\chi,1}^f)$ we have $\rho_+\Psi\in D(U_{\oS_1,+})$.
It follows that $\R$ is well-defined.

To conclude our proof, it is enough to notice that the M\o ller operator is a composition of isomorphisms.
As such the inverse $\R_{1,0}^{-1}$ of $\R_{1,0}$ can be computed explicitly as $\R_{1,0}^{-1}=\bkappa_{0,1}\circ U_{\oS_{0,1}^f,-}\circ\rho_-\circ U_{\oS_{\chi,1}^f+}\circ\rho_+$.
\end{proof}

\begin{example}
Let $(\M,g)$ be a globally hyperbolic spacetime with timelike boundary and let 
$\oS$ and $\overline{\oS}$ be   symmetric weakly-hyperbolic systems of constant 
characteristic which differ by a zero order term, \ie $\oS-\overline{\oS}=V$, 
for $V\in\Gamma(\End(\E))$. It follows that any  self-adjoint  admissible boundary condition 
$\oB$ for $\oS$ is also a  self-adjoint  admissible boundary condition for $\overline{\oS}$. 
Therefore we can chose $\oB_\chi=\oB$.
\end{example}

{We conclude this section by showing that for any pair of admissible boundary 
conditions $\oB,\oB'$ for a given symmetric weakly-hyperbolic system there 
exists an interpolating admissible boundary condition $\oB_\chi$.
In case $\oB$ and $\oB'$ are self-adjoint and the interpolating admissible boundary condition can be constructed to be self-adjoint, then this applies in particular to Theorem \ref{thm:Moller} for the choices $V=\E_1$, $W_0=\bkappa_{1,0}(\oB_0)$, $W_1=\oB_1$.}

\begin{lemma}\label{l:deformadmissbc}
Let $V\to M$ be any smooth vector bundle over any smooth manifold $M$.
Let $q$ be any smooth quadratic form on $V$.
Assume $k:=\dim\ker(q)$, the numbers $n_+$ of positive and $n_-$ of negative 
pointwise eigenvalues of $q$ to be constant on $M$.
Let $W_0,W_1\to M$ be any 
$n_++k$-dimensional subbundles of $V$ such that $q_{|_{W_i}}\geq0$ holds 
pointwise for both 
$i=0,1$.\\
Then there exists a smooth map  
$\phi\colon[0,1]\times W_0\to 
 V$ such that, for every $t\in[0,1]$, 
$\phi_t:=\phi(t,\cdot)$ is a linear and injective vector-bundle-map, 
$q_{|_{\phi_t(W_0)}}\geq0$ holds pointwise, and $\phi_0=\mathrm{Id}_{W_0}$ as 
well as 
$\phi_1(W_0)=W_1$ are satisfied.
\end{lemma}

\begin{proof}
If we can find a smooth subbundle $W_0'$ of $V$
such that $W_0\oplus W_0'=W_1\oplus W_0'=V$ and $q_{|_{W_0'}}\leq0$ pointwise, 
then the map $\phi$ can be 
constructed as follows.
Let $\pi_{W_0}$ (\textit{resp.} $\pi_{W_0'}$) be the pointwise linear projection onto 
$W_0$ with 
kernel $W_0'$ (\textit{resp.} onto $W_0'$ with kernel $W_0$).
Then $\pi_{W_0}{}_{|_{W_1}}\colon W_1\to W_0$ is an isomorphism by $W_1\cap 
W_0'=\{0\}$ and equality of space dimensions.
Let $F:=\pi_{W_0'}\circ\left(\pi_{W_0}{}_{|_{W_1}}\right)^{-1}\colon W_0\to 
W_0'$.
Observe that, for every $v\in W_0$, we can write 
\[v+F(v)=\pi_{W_0}\left((\pi_{W_0}{}_{|_{W_1}})^{-1}(v)\right)+\pi_{W_0'}\left
((\pi_{W_0}{}_{|_{W_1}})^{-1}(v)\right)=\left(\pi_{W_0}{}_{|_{W_1}}\right)^
{-1}(v),\]
so that $v+F(v)\in W_1$.
Now define $\phi\colon[0,1]\times W_0\to V$ by $\phi(t,v):=v+tF(v)$ for all 
$(t,v)\in [0,1]\times W_0$.
Clearly $\phi$ is smooth, $\phi_t=\phi(t,\cdot)$ is a linear injective 
vector-bundle-map for every $t\in[0,1]$ because of $W_0\cap W_0'=\{0\}$ and 
obviously $\phi_0=\mathrm{Id}_{W_0}$ and $\phi_1(W_0)=W_1$ hold by the above 
observation.
Moreover, for any $(t,v)\in[0,1]\times W_0$,
\[q(v+tF(v),v+tF(v))=q(v,v)+2q(v,F(v))t+q(F(v),F(v))t^2.\]
Since the r.h.s of the last identity is a degree-$2$-polynomial in $t$, it is 
non-negative on $[0,1]$ as soon as it is for $t=0$ and $t=1$ and 
$q(F(v),F(v))\leq0$.
Therefore, if $q_{|_{W_0'}}\leq0$, then $q_{|_{\phi_t(W_0)}}\geq0$.\\
To construct $W_0'$, we make use of the following fact:\\
{\bf Lemma:} {\sl Let $A$ be any smooth section of $\mathrm{End}(V)$.
If $x\mapsto \dim\ker(A(x))$ is constant on $M$, then $\ker(A)\to M$ defines a 
smooth vector subbundle of $V$.}\\
{\it Proof:} Fix any Euclidean \textit{resp.} Hermitian inner product on $V$ 
and let $k:=\dim\ker(A(x))$ for all $x\in U$.
For any $x\in U$, we have 
$\ker(A(x))=\mathrm{ran}(A(x)^*)^\perp$, where $A(x)^*$ is the adjoint of $A(x)$ 
w.r.t. the chosen inner product on $V$.
Now $\mathrm{ran}(A^*)\to M$ defines a smooth subbundle of $V$.
Namely it defines an $(n-k)$-dimensional vector subspace of 
$V$ at each point of $M$; moreover, for any $x_0\in M$, there 
exists an open neighbourhood $U$ of $x_0$ in $M$ and a family of smooth sections
$v_1,\ldots,v_{n-k}$ of $V_{|_U}$ such that 
$\{A(x)^*v_1(x),\ldots,A(x)^*v_{n-k}(x)\}$ is a family of linearly independent 
vectors and therefore a basis of $\mathrm{ran}(A(x)^*)$ for any $x\in U$.
This shows $\mathrm{ran}(A^*)\to M$ to be a smooth subbundle of $V$.
As a straightforward consequence, its pointwise orthogonal complement must be a 
smooth subbundle as well.
This proves our claim.
\hfill\checkmark\\
It can be deduced from the claim that $\ker(q)\to M$ defines a smooth 
subbundle 
of $V$.
Therefore there exists a smooth supplementary subbundle $W$ to $\ker(q)$.
The restriction of $q$ to $W$ defines a smooth nondegenerate quadratic form.
Its signature is also constant, namely it is $(n_+,n_-)$.
By e.g. \cite[Theorem C.1.4]{Nardmannthesis2004}, the bundle $W$ can therefore 
be split as $W=W_+\oplus W_-$, where $W_\pm$ are smooth subbundles of $W$ of 
rank $n_\pm$ and on which $q$ restricts pointwise as a positive- \textit{resp.} 
negative-definite quadratic form.
On the whole, we obtain the smooth splitting $V=\ker(q)\oplus 
W_+\oplus W_-$.
Now $W_0':=W_-$ does the job since automatically $W_0\cap W_-=W_1\cap 
W_-=\{0\}$ by the fact that $q_{|_{W_-}}$ is pointwise negative definite.
This concludes the proof of Lemma \ref{l:deformadmissbc}.
\end{proof}
To apply Lemma \ref{l:deformadmissbc}, consider 
$q:=\fiber{\sigma_{\oS_1}(\n^\flat)\cdot}{\cdot}$ on $V:=\E_1{}_{|_{\bM}}$ as 
well as $W_0:=\oB_{0,1}$ and $W_1:=\oB_1$.
Then the map $\hat{\phi}$ realizing the interpolation of the boundary 
conditions is defined by 
\[\hat{\phi}\colon\oB_{0,1}\to\E_1{}_{|_{\bM}},\qquad v\mapsto 
\phi(\chi(\pi(v)),v),\]
where $\pi\colon \E_1{}_{|_{\bM}}\to\bM$ is the footpoint map.\\

\begin{remarks}
\noindent 
\begin{itemize}
\item[1.]
In case $W_0$ and $W_1$ are null spaces for $q$, then unless $q$ vanishes identically on $V$ and thus $W_0=W_1=V$ the space $\phi_t(W_0)$ as constructed in the proof of Lemma \ref{l:deformadmissbc} is not null for almost every $t\in[0,1]$.
This does not prevent the existence of a path of null subspaces connecting $W_0$ and $W_1$: namely the question is only whether the Grassmannian of $n_++n_0$-dimensional $q$-nonnegative subspaces in an $n_++n_0+n_-$-dimensional one is connected or not, where $n_0=\dim\ker\sigma_\oS(\n^\flat)$.
\item[2.] \label{r:Mollerandnonconstantop}
Note also that Lemma \ref{l:deformadmissbc} can be applied to the situation 
where a 
stronger condition as condition \eqref{assumption:S0andS1} on the operator 
$\oS_1$ is assumed, namely that the numbers 
$n_0,n_+,n_-$ of vanishing, positive resp. negative eigenvalues of 
$\sigma_{\oS_1}(\xi)$ are constant along $\bM$ whenever $\xi$ is a nonvanishing 
covector in $T^*\Sigma_{|_{\partial\Sigma}}$; then we need not assume any 
longer that $\oS_{0,1}=\oS_1$.
This applies for instance to the Dirac operators associated to two different 
globally hyperbolic metrics $g_0,g_1$ and where the boundary condition is the 
MIT one, see Section \ref{sec:MIT} below.
\end{itemize}
\end{remarks}

\subsection{Conservation of positive definite Hermitian scalar products}\label{sec:conserv}

Consider now the pre-Hilbert space given by
$$\sol_{sc,\oB}(\oS)=\{ \Psi \in \Gamma_{sc}(\E) \,|\,\oS \Psi = 0 \,,\,\,  {\Psi}_{|_\bM} \in \oB \}$$ where $\scalar{\cdot}{\cdot}$ is the positive definite Hermitian form defined by
\begin{align}\label{eq:Herm prod}
 \scalar{\cdot}{\cdot}=\int_{\Sigma}\fiber{\cdot}{\sigma_{\oS}(\n^\flat)\cdot} \vol_{\Sigma}\,, \end{align}
 where $\n=-\frac{1}{\beta}\partial_t$ is the past-directed unit normal vector to $\Sigma$ while $\n^\flat=g(\n,\cdot)=\beta dt$.
In the next lemma, we shall show that if $\oS$ is skew-adjoint then the scalar product~\eqref{eq:Herm prod} does not depend on the choice of the Cauchy hypersurface $\Sigma\subset \M$.

\begin{lemma}\label{lem:indip Sigma}
Let $\Sigma \subset \M$ be a smooth
spacelike Cauchy hypersurface with its past-oriented unit normal vector field 
$\n$ and its induced volume element $\vol_{\Sigma}$. 
Furthermore, let $\oS$ be a 
formally skew-adjoint, symmetric weakly-hyperbolic system of constant 
characteristic with self-adjoint admissible boundary condition, \ie 
$\oB_+=\oB_-$, see Definition \ref{def:selfadjointbc}.
Then
$$\scalar{\cdot}{\cdot}\colon \sol_{sc,\oB}(\oS)\times \sol_{sc,\oB}(\oS) \to \CC \qquad \scalar{\Psi}{\Phi}=\int_{\Sigma}\fiber{\Psi}{\sigma_\oS(\n^\flat) \Phi} \vol_{\Sigma}  \,,
$$
where $\n^\flat$ denotes here the future-directed unit conormal, yields a positive definite Hermitian scalar product  which does 
not depend on the choice of $\Sigma$.
\end{lemma}
\begin{proof}
The proof virtually coincides with the one of~\cite[Lemma  3.17]{CQF1}.
First note that $\text{supp}( \Psi) \cap \Sigma$ is compact since $\text{supp}( \Psi)$ is spacelike compact,
so that the integral is well-defined.
Let $\Sigma'$ be any other smooth spacelike Cauchy
hypersurface. Without loss of generality we may assume that $\Sigma\cap 
\Sigma'=\emptyset$, otherwise a third Cauchy hypersurface lying in the common 
pasts of $\Sigma$ and $\Sigma'$ has to be chosen, see proof of 
\cite[Lemma 3.17]{CQF1}. Let 
$\M_\T=t^{-1}(\tau,\tau')$ be the time strip such that $t^{-1}(\tau)=\Sigma$ and 
$t^{-1}(\tau')=\Sigma'$. Its boundary is $\bM_\T=(\bM\cap \M_\T) \cup \Sigma 
\cup \Sigma'$. By the Green identity~\cite[Lemma 2.11]{Ginoux-Murro-20} we have
$$ \int_{\M_T}(\fiber{\oS \Psi}{ \Phi} - \fiber{ \Psi}{\oS^\dagger \Phi} )\vol_{\M_T} = \int_{\bM_T}\fiber{\Psi}{\sigma_\oS(\n^\flat) \Phi} \vol_{\bM_T} $$
for any $\Psi,\Phi\in\sol_{sc,\oB}(\oS)$. Since $\oS$ is assumed to be 
skew-adjoint, the left-hand side of the latter equality vanishes identically. 
Moreover, since $\oB=\oB^\dagger$  also $\fiber{\Psi}{\sigma_\oS(\n^\flat) 
\Phi}$ vanishes identically at $\bM\cap \M_\T$. Therefore we can conclude
$$ 0 = \int_{\Sigma'}\fiber{\Psi}{\sigma_\oS(\n^\flat) \Phi} \vol_{\Sigma'} - \int_{\Sigma}\fiber{\Psi}{\sigma_\oS(\n^\flat) \Phi} \vol_{\Sigma}\,. $$
This concludes our proof.
\end{proof}

With the next lemma, we shall show that there exists a choice of $f$ which makes the operator $\oS_{\chi,1}^f\colon\Gamma_{sc,\oB_\chi}(\E_1)\to\Gamma_{sc}(\E_1)$ formally skew-adjoint on $\Gamma_{sc,\oB_\chi}(\E_1)$, provided 
that $\oB_\chi$ is a self-adjoint boundary condition and $\oS_0$ (\textit{resp.} $\oS_1$) are formally skew-adjoint with respect to the pairing $\scalar{\cdot}{\cdot}_0$ (\textit{resp}. $\scalar{\cdot}{\cdot}_1$).

\begin{proposition}\label{prop:Schi skewad}
Assume the setup of Theorem~\ref{thm:Moller}, that $\oS_0$ and 
$\oS_1$ are formally skew-adjoint with respect to the pairings $\scalar{\cdot}{\cdot}_0$ and $\scalar{\cdot}{\cdot}_1$ respectively.
Finally assume that $\oB_\chi$ is
self-adjoint boundary condition  for 
$\oS^f_{\chi,1}$. 
If
 $f \in C^\infty(\M)$ is the positive smooth function such that
 \[\vol_{\M_0}=f^2 
\vol_{\M_1}\]
on $\M$, 
where $\vol_{\M_0}$  (\textit{resp.} $\vol_{\M_1}$) is the 
volume form of the metric $g_0$ (\textit{resp.} $g_1$) on $\M$, then $\oS_{\chi,1}^f$ is a skew-adjoint operator in the Hilbert space $\mathscr{H}_1$.
\end{proposition}
\begin{proof}
First we compute the formal adjoint of $\oS_{0,1}^f$ on $(\M,g_1)$.
Let $\Psi_1,\Phi_1\in\Gamma_c(\E_1)$ be such that their supports do not meet 
$\bM$.
Since by assumption $f^2\vol_{\M_1}=\vol_{\M_0}$ and $\oS_0^\dagger=-\oS_0$, we 
have 
\begin{eqnarray*}
\int_{\M}\fiber{\oS_{0,1}^f\Psi_1}{\Phi_1}_1\vol_{\M_1}&=&\int_{\M}\fiber{
\bkappa_ {1,0}\oS_0\bkappa_{0,1}\Psi_1}{\Phi_1}_1\vol_{\M_1}\\
&=&\int_{\M}\fiber{\bkappa_
{1,0}\oS_0\bkappa_{0,1}\Psi_1}{\bkappa_{1,0}\bkappa_{0,1}\Phi_1}_1\vol_{\M_1}\\
&=&\int_{\M}f^2\fiber{\kappa_{1,0}\oS_0\bkappa_{0,1}\Psi_1}{\kappa_{1,0}\bkappa_
{0,1}\Phi_1}_1\vol_{\M_1}\\
&=&\int_{\M}f^2\fiber{\oS_0\bkappa_{0,1}\Psi_1}{\bkappa_
{0,1}\Phi_1}_0\vol_{\M_1}\\
&=&\int_{\M}\fiber{\oS_0\bkappa_{0,1}\Psi_1}{\bkappa_
{0,1}\Phi_1}_0\vol_{\M_0}\\
&=&\int_{\M}\fiber{\bkappa_{0,1}\Psi_1}{\oS_0^\dagger\bkappa_
{0,1}\Phi_1}_0\vol_{\M_0}\\
&=&\int_{\M}\fiber{\bkappa_{0,1}\Psi_1}{\bkappa_{0,1}
\bkappa_{1,0}\oS_0^\dagger\bkappa_
{0,1}\Phi_1}_0\vol_{\M_0}\\
&=&\int_{\M}f^{-2}\fiber{\Psi_1}{
\bkappa_{1,0}\oS_0^\dagger\bkappa_
{0,1}\Phi_1}_1\vol_{\M_0}\\
&=&\int_{\M}\fiber{\Psi_1}{
\bkappa_{1,0}\oS_0^\dagger\bkappa_
{0,1}\Phi_1}_1\vol_{\M_1}\\
&=&-\int_{\M}\fiber{\Psi_1}{
\oS_{0,1}^f\Phi_1}_1\vol_{\M_1},
\end{eqnarray*}
that is, $\left(\oS_{0,1}^f\right)^\dagger=-\oS_{0,1}^f$ on 
$(\M,g_1)$.
As a consequence,
\begin{eqnarray*}
\left((1-\chi)\oS_{0,1}
^f+\chi\oS_1\right)^\dagger&=&(1-\chi)\left(\oS_{0,1}^f\right)^\dagger-\sigma_{\oS_{0,1}^f}
(d(1-\chi))+\chi\oS_1^\dagger-\sigma_{\oS_1}(d\chi)\\
&=&-(1-\chi)\oS_{0,1}^f+\sigma_{\oS_{0,1
}^f}(d\chi)-\chi\oS_1-\sigma_{\oS_1}(d\chi)\\
&=&-(1-\chi)\oS_{0,1}^f-\chi\oS_1-2V,
\end{eqnarray*}
where $V$ is the zero-order operator defined as above by
$V:=\frac{1}{2}[\sigma_{\oS_1}(d\chi)-\sigma_{\oS_{0,1}^f}
(d\chi)]$.
Since $V$ is a Hermitian operator it follows that ${\oS}_{\chi,1}^f=(1-\chi)\oS_{0,1}
^f+\chi\oS_1+V$ is formally skew-adjoint.
\end{proof}

Building on Lemma~\ref{lem:indip Sigma} and Proposition~\ref{prop:Schi skewad}, we can show that $\R_{1,0}$ is a 
unitary map between $\sol_{sc,\oB_0}(\oS_0)$ and $\sol_{sc,\oB_1}(\oS_1)$.

\begin{proposition}\label{prop:conserv scal prod}
Assume the setup of Theorem~\ref{thm:Moller}, that $\oS_0$ and 
$\oS_1$ are formally skew-adjoint and that $\oB_\chi$ is a 
self-adjoint boundary condition  for 
$\oS^f_{\chi,1}$.
Let $\Sigma_1 \subset J^+(\Sigma_+)$ and $\Sigma_0\subset 
J^-(\Sigma_-)$ be fixed spacelike Cauchy hypersurfaces of $\M$ (w.r.t. $g_0$ or 
$g_1$, it makes no difference).
Let
 $f \in C^\infty(\M)$ be the positive smooth function such that \[\vol_{\M_0}=f^2 
\vol_{\M_1}\]
on $\M$, 
where $\vol_{\M_0}$  (\textit{resp.} $\vol_{\M_1}$) is the 
volume form of the metric $g_0$ (\textit{resp.} $g_1$) on $\M$.  
Then the M\o ller 
operator $\R_{1,0}\colon \sol_{sc,\oB_0}(\oS_0)\to \sol_{sc,\oB_1}(\oS_1)$ is a 
unitary map once $\sol_{sc,\oB_0}(\oS_0)$ (\textit{resp.} 
$\sol_{sc,\oB_1}(\oS_1)$) is equipped with the scalar product defined in 
Equation \eqref{eq:Herm prod} associated with $\oS_0$ (\textit{resp.} with 
$\oS_1$).
\end{proposition}
\begin{proof}
Let be $\Psi_0,\Phi_0\in\sol_{sc,\oB_0}(\oS_0)$ and 
$\Psi_1:=\R_{1,0}(\Psi_0),\Phi_1:=\R_{1,0}(\Phi_0)\in\sol_{sc,\oB_1}(\oS_1)$, 
where the M\o ller operator $\R_{1,0}$ is defined using the interpolating 
operator ${\oS}_{\chi,1}^f$ instead of $\oS_{\chi,1}^f$.
We also denote by $\Psi_{\chi,1}$ (resp. $\Phi_{\chi,1}$) the smooth section
with spacelike compact support in 
$\ker\left({\oS}_{\chi,1}^f\right)$ on $\M$ with 
$\Psi_{\chi,1}{}_{|_{\Sigma_-}}=\bkappa_{1,0}\Psi_0{}_{|_{\Sigma_-}}$ (resp. 
$\Phi_{\chi,1}{}_{|_{\Sigma_-}}=\bkappa_{1,0}\Phi_0{}_{|_{\Sigma_-}}$).
By Lemma \ref{lem:indip Sigma} and definition of $f$, we have  
$\n_0^\flat\otimes\vol_{\Sigma_-,g_0}=f^2 
\n_1^\flat\otimes\vol_{\Sigma_-,g_1}$ along $\Sigma_-$ and 
therefore
\begin{eqnarray*}
\int_{\Sigma_0}\fiber{\sigma_{\oS_0}(\n_0^\flat)\Psi_0}{\Phi_0}_0\vol_{\Sigma_0,
g_ 0} 
&=&\int_{\Sigma_-}\fiber{\sigma_{\oS_0}(\n_0^\flat)\Psi_0}{\Phi_0}_0\vol_{
\Sigma_- ,g_0}\\
&=&\int_{\Sigma_-}f^{-2}
\fiber { \kappa_{1,0}^f 
\sigma_{\oS_0}(\n_0^\flat)\kappa_{0,1}^f\kappa_{1,0}^f\Psi_0}{\kappa_{1,0}
^f\Phi_0 } _1\vol_ { \Sigma_-,g_0 } \\
&=&\int_{\Sigma_-}f^{-2}\fiber{\sigma_{\oS_{0,1}^f}(\n_0^\flat)\Psi_{\chi,1}}{
\Phi_{\chi,1}}
_1\vol_ { \Sigma_- , g_
0} \\
&=&\int_{\Sigma_-}\fiber{\sigma_{\oS_{0,1}^f}(\n_1^\flat)\Psi_{\chi,1}}{\Phi_{
\chi,1}} _1\vol_{
\Sigma_- , g_1 }\\
&=&\int_{\Sigma_-}\fiber{\sigma_{\oS_{\chi,1}^f}(\n_1^\flat)\Psi_{\chi,1}}{
\Phi_{\chi,1}}
_1\vol_ {
\Sigma_- ,g_1}\\
&=&\int_{\Sigma_+}\fiber{\sigma_{{\oS}_{\chi,1}^f}(\n_1^\flat)\Psi_{\chi,1}}
{\Phi_{\chi,1}
}
_1\vol_{
\Sigma_+,g_1}\\
&=&\int_{\Sigma_+}\fiber{\sigma_{\oS_1}(\n_1^\flat)\Psi_1}{\Phi_1}
_1\vol_{
\Sigma_+,g_1},
\end{eqnarray*}
which concludes the proof of Proposition \ref{prop:conserv scal prod}.
\end{proof}

\begin{definition}
We call {\it unitary M\o ller operator} the operator $\R_{1,0}$ defined in accordance with Proposition~\ref{prop:conserv scal prod}.
\end{definition}

We conclude this section with the following Remark.
\begin{remark}\label{rmk:dec}
The unitary M\o ller operator $\R_{1,0}: \sol_{sc,\oB_0}(\oS_0)\to \sol_{sc,\oB_1}(\oS_1)$ can be seen as the composition of two unitary M\o ller operators 
\begin{align*}
&\R_{\chi,0}\colon\sol_{sc,\oB_0}(\oS_0)\to \sol_{sc,\oB_\chi}(\oS_{\chi,1}^f)
\qquad
\R_{\chi,0}:=U_{\oS_{\chi,1}^f,-}\circ\rho_-\circ\bkappa_{1,0}\,,
\\&\R_{1,\chi}\colon\sol_{sc,\oB_\chi}(\oS_{\chi,1}^f) \to \sol_{sc,\oB_1}(\oS_1)
\qquad
\R_{1,\chi}:=U_{\oS_1,+}\circ\rho_+\,.
\end{align*}
\end{remark}

\section{The algebraic approach to quantum Dirac fields}\label{sec:AQFT}

In this section we shall compare the quantization of the Dirac field on two different (yet related) globally hyperbolic spacetimes with timelike boundary.
To this end, we shall benefit from~\cite{Moller,simo3,defarg}, where a class of M\o ller operator was introduced in order to construct unitary equivalent quantum field theories, together of the results of the previous Sections \ref{sec:SWHS}-\ref{sec:Moller}.

As a first step we introduce the relevant geometrical objects, showing how they fit within the framework introduced in Section \ref{sec:Moller}.
In particular we shall apply Theorem \ref{thm:Moller} and Proposition \ref{prop:conserv scal prod} for the case of the Dirac operator with MIT boundary conditions ---\textit{cf.} Equation \eqref{eq: MIT bc}.

\subsection{The Dirac operator}\label{sec:Dirac}

Let $(\M,g)$ be a globally hyperbolic manifold and assume to have a spin 
structure 
\ie  a twofold covering map from the $\Spin_0(1,n)$-principal bundle $\P_{\rm 
Spin_0}$  to the bundle of positively-oriented tangent frames $\P_{\rm SO^+}$ of 
$\M$ such that the following diagram is commutative:
\begin{flalign*}
\xymatrix{
\P_{\Spin_0} \times \Spin_0(1,n) \ar[d]_-{} \ar[rr]^-{} && \P_{\Spin_0} 
\ar[d]_-{} 
\ar[drr]^-{}   \\
\P_{\rm SO^+} \times \textnormal{SO}(1,n)   \ar[rr]^-{} &&  \P_{\rm SO^+}  
\ar[rr]^-{} &&  \M \,.
}
\end{flalign*}
\begin{remark}\label{rmk:metric dep}
Note that unlike differential forms, the definition of a spin structure depends on the metric of the underlying manifold. 
\end{remark}
The existence of spin structures is related to the topology of $\M$.
A 
sufficient (but not necessary) condition for the existence of a spin structure 
is the parallelizability of the manifold.
Therefore, since any $3$-dimensional 
orientable manifold is parallelizable, it follows by Theorem~\ref{thm: Sanchez} 
that any 4-dimensional globally hyperbolic manifold admits a spin structure.
Given a fixed spin structure, one can use the spinor representation to 
construct 
the {spinor bundle}, \ie the complex vector bundle
$$\S\M:=\Spin_0(1,n)\times_\rho \CC^N$$
where $\rho: \Spin_0(1,n) \to \textnormal{Aut}(\CC^N)$ is the complex 
$\Spin_0(1,n)$ representation and $N:= 2^{\lfloor \frac{n+1}{2}\rfloor}$. 
The spinor bundle comes together with the following structures:
\begin{itemize}
\item[-] a natural $\Spin_0(1, n)$-invariant indefinite fiber metrics
\begin{equation*}\label{eq: spin prod}
\fiber{\cdot}{\cdot}_p: \S_p\M \times \S_p\M \to \CC\,;
\end{equation*}
\item[-] a \textit{Clifford multiplication}, \ie  a fiber-preserving map 
$$\gamma\colon \T\M\to \text{End}(\S\M)\,;$$ 
which satisfies
 for all $p \in \M$, $u, v \in \T_p\M$ and $\psi,\phi\in \S_p\M$
\begin{equation*} \label{eq:gamma symm}
 \gamma(u)\gamma(v) + \gamma(v)\gamma(u) = -2g(u, v)\Id_{\S_p\M}\, \quad 
\text{and}\quad \fiber{\gamma(u)\psi}{\phi}_p=\fiber{\psi}{\gamma(u)\phi}_p\,.
\end{equation*}
\end{itemize}
Using the spin product~\eqref{eq: spin prod}, we denote as \textit{adjunction map}, the complex anti-linear vector bundle isomorphism by
\begin{equation}\label{eq:adj map}
\Upsilon_p:\S_p\M_g\to \S^*_p\M_g  \qquad \psi \mapsto \fiber{\psi}{\cdot}\,,
\end{equation}
where  $\S^*_p\M_g$ is the so-called \textit{cospinor bundle}, \ie  the dual bundle of $\S_p\M_g$.
\begin{definition} 
The \textit{(classical) Dirac operator} $\Dir$ is the operator defined as the 
composition of the metric connection $\nabla^\S$ on $\S\M$, obtained as a lift 
of the Levi-Civita connection on $\T\M$, and the Clifford multiplication:
$$\Dir=\gamma\circ\nabla^{\S\M} \colon \Gamma(\S\M) \to \Gamma(\S\M)\,.$$
\end{definition}

In local coordinates and with a trivialization of the spinor bundle $\S\M$, the 
Dirac operator reads as
\begin{align*}
\Dir \psi = \sum_{\mu=0}^{n}  \epsilon_\mu \gamma(e_\mu) \nabla^{\S\M}_{e_\mu} 
\psi\,,
\end{align*}
where  $\{e_\mu\}$ is a local Lorentzian-orthonormal frame of $\T\M$ and 
$\epsilon_\mu=g(e_\mu,e_\mu)=\pm 1$. 
\begin{proposition}[\protect{\cite[Proposition 6.2]{Ginoux-Murro-20}}]
 The classical Dirac operator $\Dir$ on globally hyperbolic 
spin manifolds $\M$ with timelike boundary is a nowhere characteristic symmetric 
hyperbolic system.
\end{proposition}

It follows, by Theorem~\ref{thm:SWHS}, that the Cauchy problem for the Dirac operator on globally hyperbolic spacetimes with empty boundary is well-posed, therefore, it admits a Cauchy evolution operator $U_t\colon \Gamma_{c}(\S\M|_{\Sigma_t})\to\Gamma_{sc}(\S\M)$. Remarkably, as shown by Capoferri and Vassilliev~\cite{CapProp2}, the Cauchy evolution operator for Dirac fields on Cauchy-compact ultrastatic manifolds (with empty boundary) can be realized as a Fourier integral operator. As a matter of fact, the Fourier integral representation of the propagator contains the information on how singularities propagates in the manifolds. 
For this reason, it would be desirable to extend their techniques to more general globally hyperbolic manifolds with possibly not empty boundary.

\subsection{Self-adjoint admissible boundary conditions}\label{sec:MIT}
The aim of this section is to recast the boundary conditions 
for the Dirac 
operator which are self-adjoint and admissible in the sense of Definition \ref{def:admissible bc}. Let us remark that not all physical interesting boundary conditions for Dirac fields enter in this class of boundary condition. Indeed there exists physically interesting non-local boundary conditions, like the so-called APS boundary condition, which guarantees that the Cauchy problem is well-posed~\cite{DiracAPS}, but they are not admissible (since any admissible boundary condition is a local condition). For further details on self-adjoint admissible boundary conditions for the Dirac fields  we refer to~\cite[Section 6.1.1]{Ginoux-Murro-20} and~\cite[Remark 3.19]{DiracMIT}.\medskip

The first example of self-adjoint admissible boundary conditions are the so-called \textit{chiral boundary conditions}. They are defined as follows: let $\mathcal{G}$ be a chirality operator on 
$\mathbb{S}\M$, \ie a parallel involutive antiunitary (with respect to 
$\fiber{\cdot}{\cdot}$) endomorphism-field of $\mathbb{S}\M$ that anti-commutes 
with Clifford multiplication by  vectors.
Notice that chirality operators exist only in even-dimensional manifolds.
Then the so-called \textit{chirality} boundary spaces $\mathsf{B}_{\textsc{chi}}$ are defined as 
the range of the maps
\begin{align}\label{eq: Chiral bc}
\pi^\pm_{\rm CHI}
:=\frac{1}{2}\left(\mathrm{Id}\pm\gamma(\n)\mathcal{G}\right)\,,
\end{align}
where $\gamma(\n)$ denotes Clifford multiplication for the outward-pointing unit 
normal along $\bM$. It is not difficult to check that the range of $\pi_{\rm 
CHI}$ has dimension $2^{\lfloor \frac{n}{2} \rfloor -1}$, which is the number of nonnegative eigenvalues of the endomorphism $\sigma_\Dir(\n^\flat)$, and 
$$\fiber{\sigma_\Dir(\n^\flat)\pi^\pm_{\rm CHI}(\psi)}{\pi^\pm_{\rm CHI}\psi}=0\,.$$
 Furthermore, since
$$\pi^\pm_{\rm CHI}\pi^\pm_{\rm CHI}=\pi^\pm_{\rm CHI}\,, \qquad
\pi^\mp_{\rm CHI}\pi^\pm_{\rm CHI}=0\,,
\qquad 
\pi^+_{\rm CHI} + \pi^-_{\rm CHI}=\Id\,,$$
it can be easily verified that the boundary conditions are self-adjoint.
 \medskip

The second example of self-adjoint boundary conditions is the so-called \textit{MIT boundary conditions.} The boundary space $\mathsf{B}_{\textsc{mit}}$ is defined as the range of 
\begin{align}\label{eq: MIT bc}
\pi^\pm_{\rm MIT}
:=\frac{1}{2}\left(\mathrm{Id} \pm \imath\gamma(\n)\right)\,,
\end{align}
where $\gamma(\n)$ is again the Lorentzian Clifford multiplication for the 
outward-pointing unit normal vector along $\bM$. Similarly to the chiral boundary conditions, the range of $\pi_{\rm 
CHI}$ has dimension $2^{\lfloor \frac{n}{2} \rfloor -1}$, $$\fiber{\sigma_\Dir(\n^\flat)\pi^\pm_{\rm CHI}(\psi)}{\pi^\pm_{\rm CHI}\psi}=0$$ and we have
$$\pi^\pm_{\rm MIT}\pi^\pm_{\rm MIT}=\pi^\pm_{\rm MIT}\,, \qquad
\pi^\mp_{\rm MIT}\pi^\pm_{\rm MIT}=0 \,,
\qquad 
\pi^+_{\rm MIT} + \pi^-_{\rm MIT}=\Id\,.$$

\subsection{Linear isometry between spinor bundles}\label{sec:kappa}

We shall now apply the results obtained in the Sections \ref{sec:Moller} to compare the solution spaces associated with pairs of Dirac operators $\Dir_0,\Dir_1$ defined using different metrics $g_0,g_1\in\mathcal{GH}_\M$ and equipped with admissible self-adjoint boundary conditions.
In what follows $g_0,g_1\in\mathcal{GH}_\M$ are assumed to fulfil assumption (i) of Setup \ref{setup}.

As already underlined in~Remark~\ref{rmk:metric dep}, the space of spinors depends on the metric of the underlying manifold $\M_\alpha$. Therefore, an identification between spaces of sections of spinor bundles for different metrics is needed to construct a unitary M\o ller operator. This can be achieved by following~\cite[Section 5]{Baer}. \medskip

Consider a family of Lorentzian spin manifolds $\M_\lambda:=(\M,g_\lambda)$ 
with a common Cauchy temporal function, where $g_\lambda\in\mathcal{GH}_\M$ for
any $\lambda\in\RR$.
For a given nonempty interval $\I$ in $\RR$ let $\Z$ be the Lorentzian 
manifold 
$$ \Z= \I \times \M \qquad\qquad g_\Z =  d\lambda^2 + g_\lambda\,.$$
 On $\Z$ there exists a globally defined vector field which we denote as $e_\lambda:=\frac{\partial}{\partial \lambda}$.
For any $\lambda$, the spin structures on $\Z$ and $\M_\lambda\simeq\{\lambda\}\times\M$ are in one-to-one correspondence: Any spin structure on $\Z$ can be restricted to a spin structure on $\M_\lambda$ and a spin structure on $\M_\lambda$ it can be pulled back on $\Z$ -- see~\cite[Section 3 and 5]{Baer}. 
Actually, the spinor bundle $\S\M_\lambda$ on each globally hyperbolic spin manifold $\M_\lambda$ can be identified with the restriction of the spinor bundle $\S\Z$ on $\M_\lambda$, in particular $\S\M_\lambda\simeq \S\Z|_{\M_\lambda}$ if $n$ is even, while $\S\M_\lambda\simeq \S^+\Z|_{\M_\lambda}\simeq \S^-\Z|_{\M_\lambda}$ if $n$ is odd.
Equivalently we may identify
\begin{equation}\label{Eqn: isomorfism between SM on Sigma and SSigma}
\S\Z|_{\M_\lambda}=
\begin{cases}
\S\M_\lambda & \text{if $n$ is even,}\\
\S\Z|_{\M_\lambda}\oplus \S\Z|_{\M_\lambda} & \text{if $n$ is odd.}
\end{cases}
\end{equation}
By denoting with $\gamma_\Z$ (\textit{resp}. $\gamma_\lambda$) the Clifford multiplication on $\S\Z$ (\textit{resp.} on $\S\M_\lambda$), the family of Clifford multiplications $\gamma_{\lambda}$ satisfies
\begin{align}\label{eq:n even}
&\gamma_{\lambda}(v)\psi=
\gamma_\Z(e_\lambda) \gamma_\Z(v) \psi
\qquad &\text{if $n$ is even,}\\ \label{eq:n odd}
&\gamma_{\lambda}(v)(\psi_++\psi_-)=
\gamma_\Z(e_\lambda) \gamma_\Z(v) (\psi_+ - \psi_-)
\qquad &\text{if $n$ is odd,}
\end{align}
where in the second case $\psi=\psi_++\psi_-\in \S\Z|_{\M_\lambda}\oplus \S\Z|_{\M_\lambda}$ and each component $\psi_\pm$ is identified with an element in $\S^\pm \Z|_{\M_\lambda}$.

\begin{lemma}[\protect{\cite[Lemma 3.7]{defarg}}]\label{lem:cliff iso}
Let $\Z$ be the Lorentzian spin manifold given by
$$ \Z=\I \times \M \qquad\qquad g_\Z =  d\lambda^2 + g_\lambda\,,$$
 where $(\M,g_\lambda):=\M_\lambda$ is a family of Lorentzian  spin manifolds with a common Cauchy temporal function, and denote with $\S\M_{\lambda}$ be the spinor bundle over $\M_\lambda$.
For any $p\in\M_{\lambda}$,  the map 
\begin{equation}\label{Eq: isometric bundle isomorphisms}
 \kappa_{1,0}\colon \S_{p}\M_0\to \S_{p}\M_1\,.
\end{equation}
defined by the   parallel translation on $\Z$
along the curve $\lambda \mapsto (\lambda, p)$
 is a linear isometry and preserves the Clifford multiplication, \ie for any 
$v\in\Gamma(\T\M)$ and any $\Psi_0\in\Gamma(\S\M_0)$ it holds
 $$\gamma_1(\wp_{1,0} v) (\kappa_{1,0}  \Psi_0) =\kappa_{1,0} \,\big(\gamma_0(v)\Psi_0\big)\,,$$
 where $\wp_{1,0}:\T\M_0\to\T\M_1$ is the parallel transport along the curve $\lambda\mapsto(\lambda,p)$.
\end{lemma}
\begin{remark}
Let us remark, that for any couple of Lorentzian metric $g_0$ and $g_1$ admitting a common Cauchy temporal function, there always exists a path of Lorentzian metric $g_\lambda$  connecting $g_0$ to $g_1$, e.g.  $g_\lambda= \lambda g_1 + (1-\lambda)g_0$ where $\lambda\in[0,1]$. For more details we refer to~\cite{defarg,defargNor}.
\end{remark}

Lemma \ref{lem:cliff iso} provides an isomorphism $\kappa_{1,0}\colon\S\M\to\S\M$ with the same properties introduced in the Setup \ref{setup}.
We shall denote by $\Dir_{0,1}^f$ the intertwining Dirac operator as in Proposition \ref{prop: S01 is wSHS}.
Similarly $\Dir_{\chi,1}^f$ shall denote the operator interpolating between $\Dir_{0,1}^f$ and $\Dir_1$.
Here and in what follows $f$ is chosen as per Proposition \ref{prop:conserv scal prod}.

\begin{remark}\label{rmk:MIT01}
Keeping the notation of Remark~\ref{rem:symb} and Lemma~\ref{lem:cliff iso}, the diffeomorphism $\zeta:\M\to\M$ is simply the identity $\Id$.
Since $\sigma_{\Dir_0}(\xi)=\gamma_0(\xi^{\sharp_0})$, where $^{\sharp_0}$ denotes the musical isomorphism with respect to $g_0$, we find
\begin{align*}
\sigma_{\Dir_{0,1}^f}(\xi_1)
=\bkappa_{0,1}\sigma_{\Dir_0}(\xi_1)\bkappa_{0,1}
=\bkappa_{1,0}\gamma_0 (\xi_1^{\sharp_0})\bkappa_{0,1}
=\gamma_1(\wp_{1,0}\xi_1^{\sharp_0})
=\sigma_{\Dir_1} \Big((\wp_{1,0}\xi_1^{\sharp_0})^{\flat_1}\Big)
=\sigma_{\Dir_1}(\wp_{1,0}\xi_1)\,,
\end{align*}
where $\sharp_1:=\flat_1^{-1}$ is the musical isomorphism associated with $g_1$.
In the last equality we used that, for $\xi\in T_x^*\M$ and $X\in T_x\M$ we have
\begin{align*}
(\wp_{1,0}\xi^{\sharp_0})^{\flat_1}(X)|_x
&=g_1(\wp_{1,0}\xi^{\sharp_0},X)|_x
=g_\Z(\wp_{1,0}\xi^{\sharp_0},X)|_{(1,x)}
=g_\Z(\xi^{\sharp_0},\wp_{0,1}X)|_{(0,x)}
\\&=g_0(\xi^{\sharp_0},\wp_{0,1}X)|_x
=\xi(\wp_{0,1}X)|_x
=[\wp_{0,1}\xi](X)|_x\,,
\end{align*}
where, with a slight abuse of notation, we denoted $\wp_{0,1}\xi$ the parallel transport of the $1$-form $\xi$ along the curve $\lambda\to(\lambda,x)$ within $\Z$: The latter coincides with $\wp_{0,1}^*\xi$, being $\wp_{0,1}\colon \T\M_1\to \T\M_0$.
\end{remark}

We are almost in position to apply Theorem \ref{thm:Moller} and Proposition \ref{prop:conserv scal prod}.
In the next lemma we shall prove that the assumption in Theorem~\ref{thm:Moller} that the parallel transport of 
$\n_1^\flat$ is not proportional to $\mu\n_1^\flat$ for any $\mu<0$, is always 
satisfied, provided $g_\lambda=(1-\lambda)g_0+ \lambda g_1$ for all 
$\lambda\in[0,1]$. 

\begin{lemma}\label{lem:wpn0}
Let $(\M,g_0)$ and $(\M,g_1)$ be globally hyperbolic manifolds with timelike 
boundary split as $(\M,g_i)=(\mathbb{R}\times\Sigma,-\beta_i^2dt^2\oplus h_{i}(t))$ for 
both $i=0,1$.
Consider the manifold $\Z:=[0,1]\times \M$ endowed with the metric 
$g_\Z:=d\lambda^2\oplus g_\lambda$, where
\begin{align*}
g_\lambda:=
(1-\lambda)g_0
+\lambda g_1
=-\beta_\lambda dt^2\oplus h_\lambda(t)\,,
\end{align*}
where $\beta_\lambda:=(1-\lambda)\beta_0+\lambda\beta_1$ and $h_\lambda(t)=(1-\lambda)h_0(t)+\lambda h_1(t)$.
Then $h_1(\wp_{1,0}(\n_1),\n_1)>0$ along $\bM$, where $\wp_{1,0}$ is the 
parallel transport in 
$(\Z,\hat{g})$ along $[0,1]\to\Z$, $\lambda\mapsto(\lambda,p)$, for any 
$p\in\bM$.
\end{lemma}

\begin{proof}
Note that, by definition of both $g_i$ and of $g_\Z$, we have 
$\nabla^\Z_{\partial_\lambda}\partial_\lambda=\nabla^\Z_{\partial_\lambda}
\beta_\lambda^{-1}\partial_t=0$ , so that, for any $\lambda_0\in[0,1]$, the parallel transport 
along $[0,\lambda_0]\to\Z$, $\lambda\mapsto(\lambda,p)$ preserves $T\Sigma$.
Writing $p=(t,x)$, we fix a pointwise $h_0$-o.n.b. of $T_x\Sigma$ in which 
$h_1=h_1(t)$ is diagonal \ie, there exist $\mu_1,\ldots,\mu_n>0$ such that $h_1(e_i,e_j)=\mu_i\delta_{ij}$ for all $1\leq i,j\leq n$.
This basis $(e_i)_{1\leq i\leq n}$ is extended constantly in $\lambda$ along 
$\lambda\mapsto(\lambda,p)$.
Splitting $\wp_{\lambda,0}\n_1=\sum_{j=1}^n\alpha_j e_j$, where 
$\alpha_j=h_0(\wp_{\lambda,0}\n_1,e_j)$, we have 
\begin{eqnarray*}
0&=&\nabla^\Z_{\partial_\lambda}\left(\wp_{\lambda,0}\n_1\right)\\
&=&\sum_{j=1}^n(\partial_\lambda\alpha_j)e_j+\alpha_j\nabla^\Z_{
\partial_\lambda}e_j\\
&=&\sum_{j=1}^n(\partial_\lambda\alpha_j)e_j+\alpha_j\left(\underbrace{[
\partial_\lambda , e_j]}_{0}+\frac{1}{2}h_\lambda^{-1}\partial_\lambda 
h_\lambda(e_j,\cdot)\right),
\end{eqnarray*}
so that, denoting by 
$Y(\lambda):=\left(\begin{array}{c}
\alpha_1(\lambda)\\\vdots\\\alpha_n(\lambda)\end{array} \right)$ and 
identifying $h_\lambda$ (as a homomorphism $T\Sigma\to T^*\Sigma$) and 
$\partial_\lambda h_\lambda$ (as symmetric $2$-tensor on $T\Sigma$) with their 
respective matrices $H_\lambda$ and $\partial_\lambda H_\lambda$ in the bases 
$(e_j)_{1\leq j\leq n}$ and $(e_j^*)_{1\leq j\leq n}$ respectively, the 
vector-valued function $Y$ must satisfy the linear first-order ODE 
\begin{equation}\label{eq:ode1Y}
Y'(\lambda)+\frac{1}{2}H_\lambda^{-1}\partial_\lambda H_\lambda\cdot 
Y(\lambda)=0
\end{equation}
on $[0,1]$.
In case $[H_\lambda,\partial_\lambda H_\lambda]=0$ is fulfilled for all 
$\lambda$, equation \eqref{eq:ode1Y} can be solved explicitely, namely 
$Y(\lambda)=H_\lambda^{-\frac{1}{2}}\cdot Y(0)$ for all $\lambda\in[0,1]$ is 
the solution with initial condition $Y(0)\in\mathbb{R}^n$.
But with $h_\lambda=(1-\lambda)h_0+\lambda h_1$, we have 
$H_\lambda=(1-\lambda)\mathrm{I}_n+\lambda\mathrm{diag}(\mu_1,\ldots,\mu_n)$, 
so that $\partial_\lambda 
H_\lambda=\mathrm{diag}(\mu_1,\ldots,\mu_n)-\mathrm{I}_n$ and therefore 
$[H_\lambda,\partial_\lambda H_\lambda]=0$ holds for all $\lambda\in[0,1]$.
This implies that 
\[Y(\lambda)=H_\lambda^{-\frac{1}{2}}\cdot 
Y(0)=\mathrm{diag}\left((1-\lambda+\lambda\mu_1)^{-\frac{1}{2}},\ldots,
(1-\lambda+\lambda\mu_n)^{-\frac{1}{2}}\right)\cdot Y(0)\]
holds for all $\lambda\in[0,1]$.
As a consequence, 
$Y(1)=\mathrm{diag}\left(\mu_1^{-\frac{1}{2}},\ldots,
\mu_n^{-\frac{1}{2}}\right)\cdot Y(0)$, from which 
\[h_1(\wp_{1,0}\n_1,\n_1)=H_1(Y(1),Y(0))=\sum_{j=1}^n\mu_j^{\frac{1}{2}}
\alpha_j(0)^2>0\]
and the claim follows.
\end{proof}

We conclude this section by stating Theorem \ref{thm:Moller} and Proposition \ref{prop:conserv scal prod} for the particular case of MIT boundary conditions.
 
\begin{proposition}\label{prop: moeller for dirac with MIT bc}
Let assume $g_0,g_1\in\mathcal{GH}_\M$ fulfils (i) in Setup \ref{setup}.
Let $\M_0$ (\textit{resp.} $\M_1$) be a globally hyperbolic spin manifold with timelike boundary and let $\Dir_0$ (\textit{resp.} $\Dir_1$) be a classical Dirac operator coupled with MIT boundary condition $\oB_{{\rm MIT}_0}$ (\textit{resp.} $\oB_{{\rm MIT}_1}$). Then the boundary space defined by 
$$\oB_\chi \:=\ker M_\chi:=\ker \Big( \gamma_1( v) - \imath \|v\|_1 \Big)$$ is a self-adjoint boundary space for the operator 
$$\Dir^f_{\chi,1}:=(1- \chi) \,\bkappa_{1,0}\Dir_0 \bkappa_{0,1} + \chi\Dir_{1} + \frac{1}{2} \left( \sigma_{\Dir_1} + \sigma_{\Dir_{0,1}^f}\right)(d\chi)\,,$$ 
where $v= \chi \n_1 + (1-\chi) \wp_{1,0} \n_1$ and $\|v\|_1=\sqrt{g_1( v,v)}$.

Therefore, letting $\sol_{sc,\textsc{mit}}(\Dir_i):=\{\Psi\in\S\M_i\,|\,\Dir_i\Psi=0\,,\,\Psi|_\bM\in\mathsf{B}_{\textsc{mit}}\}$, there exists a unitary isomorphism (M\o ller operator) $\R_{1,0}\colon \sol_{sc,\textsc{mit}}(\Dir_0)\to\sol_{sc,\textsc{mit}}(\Dir_1)$ where $\sol_{sc,\textsc{mit}}(\Dir_0)$ (\textit{resp.} $\sol_{sc,\textsc{mit}}(\Dir_1)$) is equipped with the scalar product defined in Equation \eqref{eq:Herm prod} associated with $\Dir_0$ (\textit{resp.} with $\Dir_1$).
\end{proposition} 
\begin{proof}
The last part of the statement is nothing but Theorem \ref{thm:Moller} together with Proposition \ref{prop:conserv scal prod}.
Thus, it remains to prove the first part.
We begin by noticing that when $\chi=0$ then $\oB_\chi$ reduces to  $\ker (\gamma_1(\wp_{1,0}\n_1)-\imath) = \kappa_{1,0}(\oB_{\rm{MIT}_0})$ on account of Remark~\ref{rmk:MIT01}, while when $\chi=1$ $\oB_\chi$ reads as $\oB_{{\rm MIT}_1}$. To conclude we need to show that $\oB_\chi$ is a self-adjoint admissible boundary condition for $\Dir_{\chi,1}^f$.
To this end, we first notice that  the bilinear form $\fiber{\gamma_1(v) \Psi}{\Psi}_q$ satisfies simultaneously 
\begin{align*}
 \fiber{\gamma_1(v) \Psi}{\Psi}_q &= \fiber{ \Psi}{ \gamma_1(v)\Psi}_q  =  \overline{\fiber{\gamma_1(v) \Psi}{\Psi}_q } \\ \fiber{\gamma_1(v) \Psi}{\Psi}_q &= \fiber{\imath\|v\| \Psi}{\Psi}_q= \fiber{ \Psi}{-\imath\|v\| \Psi}_q = - \fiber{ \Psi}{ \gamma_1(v)\Psi}_q  = - \overline{\fiber{\gamma_1(v) \Psi}{\Psi}_q } 
\end{align*}
which implies that $\fiber{\gamma_1(v) \Psi}{\Psi}_q=\fiber{\gamma_{\Dir_{\chi,1}^f}(\n^\flat_1) \Psi}{\Psi}_q =0$ \,. 
Furthermore the range of the projector $\pi=\frac{1}{2}(\Id - \imath\|v\|_1^{-1} \gamma_1(v))$ which is equal to $\oB_\chi$,  has dimension  $2^{\lfloor \frac{n}{2}-1 \rfloor}$, which is exactly the number of nonnegative eigenvalues of $\gamma_{\Dir_{\chi,1}^f}(\n^\flat)$. This concludes our proof.
\end{proof}

\begin{remark}\label{rem:existBchiDiracMIT}
Since, for any nonzero spacelike covector $v$ on $\M$, the operator 
$\sigma_{\Dir}(v)$ has vanishing kernel and $\pm|v|$ as nonvanishing 
eigenvalues, each with multiplicity $2^{\left[\frac{n+1}{2}\right]-1}$, the 
existence of an interpolating $\oB_\chi$ between 
$\oB_{\textsc{mit}_0}$ and $\oB_{\textsc{mit}_1}$ for $\Dir_0$ and $\Dir_1$ 
respectively follows from Lemma \ref{l:deformadmissbc}, see Remark 
\ref{r:Mollerandnonconstantop} above.
Note however that the interpolating $\oB_\chi$ from Lemma \ref{l:deformadmissbc} is not self-adjoint.
\end{remark}

\subsection{The algebra of Dirac fields  with MIT boundary condition} \label{sec:CAR}

In this section we shall exploit Proposition \ref{prop: moeller for dirac with MIT bc} to compare the quantization of the Dirac field with MIT boundary condition on $\M_0$ and $\M_1$.
With this purpose we shall briefly recall the quantization procedure from the algebraic point of view.

In \cite{Dappia4,simo3,defarg}, the quantization of a free field theory is realized as a two-step procedure.
On the one hand, the physical system classically described by $\sol_{sc,\textsc{mit}}(\Dir)$ is quantized by introducing a unital $*$-algebra $\mathfrak{A}$, whose elements are interpreted as observables for the system under investigation.
In a second stage, the description of possible physical states of the system is described through the choice of a suitable subclass of linear, positive and normalized functionals $\omega\colon\mathfrak{A}\to\mathbb{C}$.

By extending the analogous definition for a spacetime without boundary, we shall now introduce the $\ast$-algebra $\mathfrak{A}$ associated with the space $\sol_{sc,\textsc{mit}}(\Dir)$ of solutions with spatially compact support of the Dirac operator $\Dir$ coupled with MIT boundary conditions and endowed with the positive definite Hermitian scalar product~\eqref{eq:Herm prod}.

To this avail we shall profit of the results and definition already present in the literature, see \cite{araki}.
For later convenience let $\sol_{sc,\textsc{mit}}^\oplus$ be the Hilbert space obtained by completion of
\begin{align*}
\sol_{sc,\textsc{mit}}(\Dir)\oplus \Upsilon\sol_{sc,\textsc{mit}}(\Dir)\,,
\end{align*}
equipped with the natural scalar product $(\,,\,)_{\sol_{sc,\textsc{mit}}^\oplus}$ induced by $\sol_{sc,\textsc{mit}}(\Dir)$ ---\textit{cf.} Equation \eqref{eq:Herm prod}--- in particular $\scalar{\psi_1}{\psi_2}=\int_\Sigma\fiber{\psi_1}{\gamma(-\beta^{-1}\partial_t)\psi_2}\vol_{\Sigma}$.
Moreover, let $\Gamma\colon\sol_{sc,\textsc{mit}}^\oplus\to\sol_{sc,\textsc{mit}}^\oplus$ be the antilinear involution defined by $\Gamma(\psi_1\oplus\Upsilon\psi_2):=(-\psi_2)\oplus\Upsilon\psi_1$ where $\Upsilon\colon\S\M\to\S^*\M$ has been defined in Equation \eqref{eq:adj map}.

\begin{definition}\label{def:alg Dirac}
The \textit{algebra of Dirac fields  with MIT boundary condition} is the unital, complex $*$-algebra $\mathfrak{A}$ freely generated by the abstract elements $\Xi(\psi)$, $1_{\mathfrak{A}}$, with $\psi\in\sol_{sc,\textsc{mit}}^\oplus$, together with the following relations for all $\psi,\phi\in\sol_{sc,\textsc{mit}}^\oplus$ and $\alpha,\beta\in\CC$:
\begin{itemize}
\item[(i)] Linearity: $\Xi(\alpha \psi + \beta \phi) =\alpha \Xi(\psi) + \beta\Xi(\phi)$
\item[(ii)] Hermiticity: $\Xi(\psi)^*=\Xi(\Gamma\psi)$
\item[(iii)] Canonical anti-commutation relations (CARs):
$$ \Xi(\psi) \cdot\Xi(\phi) + \Xi(\phi)\cdot \Xi(\psi) =0 \qquad \text{and} \qquad  \Xi(\psi) \cdot\Xi(\phi)^* + \Xi(\phi)^*\cdot \Xi(\psi) = \scalar{\psi}{\phi} \,1_\fA \,.$$
\end{itemize}
\end{definition}

As a matter of fact $\fA$ can be completed in a unique way into a $C^*$-algebra \cite{araki} the $C^*$-norm being induced by the natural Hilbert structure of $\sol_{sc,\textsc{mit}}^\oplus$.
Occasionally we shall implicitly regard $\fA$ as a $C^*$-algebra.

Recollecting the results of the previous sections we have the following:

\begin{theorem}\label{thm:alg iso}
Assume that $g_0,g_1\in\mathcal{GH}_\M$ fulfils (i) in the Setup \ref{setup} and let $\fA_\alpha$ be the algebra of Dirac fields with MIT boundary condition on $\M_\alpha$. Then 
 the unitary M\o ller operator $\R_{1,0}:\sol(\Dir_0)\to \sol(\Dir_1)$ lifts to 
a $*$-isomorphism $\fR_{1,0}:\fA_0\to\fA_1$.
\end{theorem}
\begin{proof}
Let $\Upsilon_\alpha:\S\M_\alpha\to\S^*\M_\alpha$ the adjunction map defined in~\eqref{eq:adj map} between the spinor and cospinor bundle over $\M_\alpha$ and set $\R_{1,0}^\Upsilon:=\Upsilon_1 \R_{1,0}\Upsilon_0^{-1}$.
Then $\R_{1,0}^\Upsilon$ implements an isomorphism between $\Upsilon_0\sol_{sc,\textsc{mit}}(\Dir_0)$ and $\Upsilon_1\sol_{sc,\textsc{mit}}(\Dir_1)$.
On account of Proposition \ref{prop: moeller for dirac with MIT bc} $\R_{1,0}^\oplus:=\R_{1,0}\oplus\R_{1,0}^\Upsilon\colon\sol_{sc,\textsc{mit}}^\oplus\to\sol_{sc,\textsc{mit}}^\oplus$ is a unitary isomorphism.
By direct inspection, the linear map $\fR_{1,0}\colon\fA_0\to\fA_1$ defined by $\fR_{1,0}\Xi(\psi):=\Xi(\R_{1,0}^\oplus\psi)$ extends to the seen $*$-isomorphism.
\end{proof}

\begin{remark}
The algebra of Dirac fields  with MIT boundary condition cannot be considered as an algebra of observables, since observables are required to commute at spacelike separations and $\fA$ does not fulfil such requirement.
A good candidate as algebra of observables is the subalgebra $\fA_{\text{obs}} \subset\fA$ consisting of elements which are even, \ie invariant by replacement $\Xi(\psi)\mapsto-\Xi(\psi)$, and invariant under the action of $\Spin_0(1,n)$ (extended to $\fA$).
For further details we refer to~\cite{DHP}.
\end{remark}

\subsection{Hadamard states}\label{sec:Hadam}

In this section we study (algebraic) states and their interplay with the $\ast$-isomorphism $\mathfrak{R}_{1,0}$.
\begin{definition}
Given a complex $*$-algebra $\mathfrak{A}$ we call  \textit{(algebraic) state} any linear functional from $\mathfrak A $ into $\CC$ that is positive, \ie $\omega(\aa^*\aa)\geq 0$ for any $\aa\in\fA $, and normalized, \ie $\omega(1_\fA)=1$. 
\end{definition}

Due to the natural grading on the algebra of Dirac fields  with MIT boundary conditions $\fA$, it suffices to define $\omega$ on the monomials.
Among all states, the so-called quasi-free states play a distinguished role.

\begin{definition}\label{def:quasifree}
A state $\omega$ on $\fA$ is \textit{quasifree} if it satisfies
\begin{align*}
\omega(\Xi(\psi_1)\cdots\Xi(\psi_n))
=\begin{cases}
0\qquad n\textrm{ odd}\\
\sum\limits_{\sigma \in S'_n} (-1)^{\text{\rm{sign}}(\sigma)} \prod\limits_{i=1}^{n/2}
\;\omega\left(
\Xi(\psi_{\sigma(2i-1)})
\Xi(\psi_{\sigma(2i)}) \right)\quad n \textrm{ even}
\end{cases}\,,
\end{align*}
where~$S'_n$ denotes the set of ordered permutations of $n$ elements.
\end{definition}
As shown in \cite[Lemma 3.2]{araki}, for any quasi-free state $\omega$ on the $C^*$-algebra $\fA$ there exists a bounded operator $Q_\omega\in B(\sol_{sc,\textsc{mit}}^\oplus)$ on $\sol_{sc,\textsc{mit}}^\oplus$ such that $0\leq Q_\omega=Q_\omega^*\leq 1$, $Q_\omega+\Gamma Q_\omega\Gamma=\operatorname{Id}_{\sol_{sc,\textsc{mit}}^\oplus}$ and
\begin{align}\label{Eq: 2-point as an operator}
\omega(\Xi(\psi_1)^*\Xi(\psi_2))=(\psi_1,Q_\omega\psi_2)_{\sol_{sc,\textsc{mit}}^\oplus}\,.
\end{align}
From a different perspective, we can realize $\omega(\Xi(\psi_1)^*\Xi(\psi_2))$ in terms of distributions.
This turns out to be quite useful when looking for physically relevant states.
To this avail we observe that, an application of Proposition \ref{prop:Green} 
leads to 
$\sol_{sc,\textsc{mit}}^\oplus\simeq\left(\rquot{\Gamma_{c}(\S\M)}{\Dir\Gamma_{
c , \textsc { mit } } (\S\M)}\right)^{\oplus 2}$ ---\textit{cf.} Equation 
\eqref{Eq: characterization of solution space}--- the isomorphism being given 
by $\left(\rquot{\Gamma_{c}(\S\M)}{\Dir\Gamma_{
c , \textsc { mit } } (\S\M)}\right)^{\oplus 2}\ni([f_1],[f_2])\to 
\mathsf{G}f_1\oplus\Gamma \mathsf{G}f_2\in\sol_{sc,\textsc{mit}}^\oplus$.
In particular we can endow $\Gamma_{c}(\S\M)$ with the standard locally convex 
topology which induces a locally convex topology on the quotient 
$\rquot{\Gamma_{c}(\S\M)}{\Dir\Gamma_{
c , \textsc { mit } } (\S\M)}$.
With this choices the map $\left(\rquot{\Gamma_{c}(\S\M)}{\Dir\Gamma_{
c , \textsc { mit } } (\S\M)}\right)^{\oplus 2}\to\sol_{sc,\textsc{mit}}^\oplus$ 
turns out to be continuous, so that to any quasi-free state we may associate its 
2-point distributions $\omega^{(2)}\in\Gamma_{c}((\S\M\oplus\S\M)^2)'$ defined 
by
\begin{align*}
\omega^{(2)}(f_1,f_2):=\omega(\Xi(\psi_{f_1})^*\Xi(\psi_{f_2}))\,.
\end{align*}
where $\psi_f\in\sol_{sc,\textsc{mit}}^\oplus$ is the element associated with 
$[f]\in\left[\rquot{\Gamma_{c}(\S\M)}{\Dir\Gamma_{
c , \textsc{mit} } (\S\M)}\right]^{\oplus 
2}$.
In particular, we have that the $2$-point distribution is a solution of the Dirac equation with MIT boundary conditions, meaning that
\begin{align}\label{Eq: dynamics for 2-point functions}
\omega^{(2)}(f_1,(\Dir\oplus\Dir) f_1)=0
\qquad\forall  f_1,f_2\in\Gamma_{c,\textsc{mit}}(\S\M\oplus\S\M)\,.
\end{align}
Notice that, due to the CAR relations, Equation \eqref{Eq: dynamics for 2-point functions} cannot be strengthened to hold true for all $f_1,\ldots,f_n\in\Gamma_{c}(\S\M\oplus\S\M)$.

A widely accepted criterion to select physically relevant states is the renowned \textit{Hadamard condition} \cite{Kay-Wald-91,Radzikowski-96,Radzikowski-Verch-96,Sahlmann-Verch-01}.
On a globally hyperbolic spacetime with empty boundary, the latter allows to construct Wick polynomials in a local and covariant fashion, moreover, it guarantees the finiteness of the fluctuations of such Wick polynomials \cite{Fewster-Verch-13}.

At a technical level, the Hadamard condition characterizes the wave front set 
$\operatorname{WF}(\omega^{(2)})\subseteq\T^*\M^2$ of the 2-point function of a 
quasi-free state ---generalization to non-quasi free states are possible 
\cite{Sanders-10}.
Such a microlocal characterization is also possible for the case of a globally hyperbolic manifold with timelike boundary: therein the Hadamard condition has been formulated in \cite{Michal} for the case of asymptotically Anti-de Sitter spacetimes and then exploited in \cite{DappiaMarta} for a wider class of boundary conditions.
In these situations the proper replacement for $\operatorname{WF}(\omega^{(2)})$ is given by $\operatorname{WF}_b(\omega^{(2)})\subset\,^b\T^*\M^2\setminus\{0\}$, where $\operatorname{WF}_b$ stands for the $b$-wave front set \cite{Melrose-93}.

\begin{definition}\label{def: Hadamard state}
Let $(\M, g)$ be a globally hyperbolic spin manifold with timelike boundary. A bidistribution $\omega^{(2)} \in \Gamma_{c}(\S\M\oplus \S\M)'$ is called of \textit{Hadamard form} if it has the
following $b$-wave front set
$$\operatorname{WF}_b(\omega^{(2)})=\{(x,y,k_x,-k_y)\in\T^*(\M\times\M)\backslash\{0\}|\ (x,k_x)\sim(y,k_y),\ k_x\rhd 0\},
$$
where $\sim$ entails that $(x,k_x)$ and $(y,k_y)$ are connected by a generalized broken bicharacteristic while that $k_x\rhd 0$ means that the covector $k_x$ at $x\in\M$ is future pointing. 
Since we deal with vector-valued distributions, the standard convention for the wave front set is to take the union of the wave front set of its components in an arbitrary but fixed local frame.
\end{definition}
  For further details on Hadamard states on globally hyperbolic manifolds with empty boundary we refer to~\cite{gerard,Gerard-Stoskopf-21-1,IgorValter}, while on globally hyperbolic manifolds with timelike boundary, we refer to~\cite{DappiaMarta, Michal,GannotWrochna}.\medskip

With the next theorem, we show that the pull-back of a quasifree state along the isomorphism $\mathfrak{R}_{1,0}:\fA_0 \to \fA_1$ induced by the unitary M\o ller operator $\R$ for $\Dir$  preserves the singularity structure of the two-point distribution $\omega^{(2)}$.

\begin{theorem}\label{thm:main appl}
Assume that $g_0,g_1\in\mathcal{GH}_\M$ fulfil (i) in the Setup \ref{setup}.
Assume furthermore that a propagation of singularities theorem holds true for $\Dir$ with MIT boundary condition, namely for any $u\in\sol_{\textsc{mit}}(\Dir)$, $\operatorname{WF}_b(u)$ is the union of maximally extended generalized broken bicharacteristics.
Denote with $\fA_\alpha$, $\alpha=0,1$, the algebras of Dirac fields with MIT boundary condition on $\M_\alpha$ and let $\omega_\alpha:\fA_\alpha\to\CC$ be quasifree states satisfying
$$\omega_0=\omega_1 \circ \mathfrak{R}_{1,0}\,: \fA_0 \to \CC$$
with $\mathfrak{R}_{1,0}$ is the isomorphism induced by $\R$ as per Theorem \ref{thm:alg iso}.
Then, if $\omega_1$ is a Hadamard state as per Definition \ref{def: Hadamard state}, then so is $\omega_0$.
\end{theorem}
\begin{proof}
Since $\fR_{1,0}$ preserves the grading of $\mathfrak{A}_0$, $\mathfrak{A}_1$, $\omega_0$ inherits the property of being a quasifree state from $\omega_1$. 
In particular the two-point function $\omega_0^{(2)}$ satisfies
\begin{align*}
\omega_0^{(2)} \left(f_0,g_0\right)
=\omega_0 \left(\Xi(\psi_{f_0})^*\Xi(\psi_{g_0})\right)
=\omega_1^{(2)} \left(\Xi(\mathsf{R}_{1,0}^\oplus\psi_{f_0})^*\Xi(\mathsf{R}_{1,0}^\oplus\psi_{g_0})\right)\,.
\end{align*}
We shall now prove that $\omega_1$ fulfils the Hadamard condition.
To this avail we first observe that $\mathsf{R}_{1,0}$ can in fact be decomposed as $\mathsf{R}_{1,0}=\mathsf{R}_{1,\chi}\circ\mathsf{R}_{\chi,0}$ (\cf Remark~\ref{rmk:dec}).
With reference to Theorem \ref{thm:Moller}, we have $\mathsf{R}_{\chi,0}:=U_{\Dir^f_{\chi,1},-}\circ\rho_-\circ\bkappa_{1,0}$ whereas $\mathsf{R}_{1,\chi}:=U_{\Dir_1,+}\circ\rho_+$.
Let us consider $\fR_{1,\chi}\colon\fA_\chi\to\fA_1$, where $\fR_{1,\chi}$ is the $*$-isomorphism defined as per Theorem \ref{thm:alg iso} with $\fA_0$ replaced with $\fA_\chi$.
Moreover let $\omega_\chi:=\omega_1\circ\fR_{1,\chi}$.
With reference to Theorem \ref{thm:Moller}, let $f_1,f_2\in\Gamma_{c}(\S\M\oplus\S\M)$ be with support contained in a neighbourhood of $\Sigma_+$.
Then
\begin{align*}
\omega_\chi^{(2)}(f_1,f_2)
&=\omega_\chi(\Xi(\mathsf{G}_\chi f_1)^*\Xi(\mathsf{G}_\chi f_2))
\qquad\textrm{def. }\omega_\chi^{(2)}
\\&=\omega_1(\Xi(\mathsf{R}_{1,\chi}^\oplus \mathsf{G}_\chi f_1)^*\Xi(\mathsf{R}_{1,\chi}^\oplus \mathsf{G}_\chi f_2))
\qquad\textrm{def. }\fR_{1,\chi}
\\&=(\mathsf{R}_{1,\chi}\mathsf{G}_\chi f_1,Q_{\omega_1}\mathsf{R}_{1,\chi}\mathsf{G}_\chi f_2)_{\sol_{sc,\textsc{mit}}(\Dir_1)}
\qquad\textrm{Eq. }\eqref{Eq: 2-point as an operator}
\\&=(\rho_-\mathsf{R}_{1,\chi}^\oplus \mathsf{G}_\chi f_1,Q_{\omega_1}\rho_-\mathsf{R}_{1,\chi}^\oplus \mathsf{G}_\chi f_2)_{\Sigma_+}
\qquad\textrm{choice of }\Sigma_+
\\&=(\rho_-\mathsf{G}_\chi f_1,Q_{\omega_1}\rho_-\mathsf{G}_\chi f_2)_{\Sigma_+}
\qquad\rho_-\mathsf{R}_{1,\chi}=\rho_-
\\&=(\rho_-\mathsf{G}_1 f_1,Q_{\omega_1}\rho_-\mathsf{G}_1 f_2)_{\Sigma_+}
\\&=\omega_1^{(2)}(f_1,f_2)\,,
\end{align*}
where we exploited the fact that, when computing $(\,,\,)_{\sol_{sc,\textsc{mit}}^\oplus}$, we may choose $\Sigma$ arbitrarily.
In the second to last equation we used that $\mathsf{G}_\chi f|_{\Sigma_+}=\mathsf{G}_1f|_{\Sigma_+}$ for $f$ supported in a small enough neighbourhood of $\Sigma_+$.
This shows that $\omega_\chi^{(2)}$ coincides with $\omega_1^{(2)}$ in a neighbourhood of $\Sigma_+$ and therefore fulfils the Hadamard condition therein.
By the assumed propagation of singularities, it follows that $\omega_\chi^{(2)}$ fulfils the Hadamard condition on $\M$.

By observing that $\omega_1=\omega_\chi\circ\fR_{\chi,0}$ and proceeding with a similar argument we have that $\omega_1$ fulfils the Hadamard condition. 
\end{proof}

\begin{remark}
We expect the propagation of singularities to hold true, as there are already positive results in this direction, see \textit{e.g.} \cite{DappiaMarta, GannotWrochna, Melrose-Sjostrand-78,Melrose-Sjostrand-82,Vasy-08} for the scalar wave equation, \cite{mosterbook,Taylor-75} for first order systems, and \cite{Baskin-Wunsch-20} for the Dirac-Coulomb system.
We postpone its investigation to a forthcoming paper.
\end{remark}

We have finally all the tools to prove the existence of Hadamard states.

\begin{proof}[Proof of Theorem~\ref{thm:Hadapl}]
Let $t$ be a Cauchy temporal function for $g$ and define $g_u:=-dt^2+h$, where $h$ is a complete Riemannian metric on $t^{-1}(s)$ for every $s\in\RR$.
On account of~\cite[Proposition 2.23]{defargNor}, there exists a globally hyperbolic metric $\overline{g}$ such that $J^+_{\overline{g}} \subset J^+_{g_u}\cap J^+_{g}$.
Denote with $\S\M_{\overline{g}}$ the spinor bundle over $(\M,\overline{g})$ and consider the linear isometries  
$$\kappa^{f'}_{\overline{g},g}\colon \S\M_g\to \S\M_{\overline{g}} \qquad \kappa^{f''}_{\overline{g},g_u}\colon \S\M_{g_u}\to \S\M_{\overline{g}}$$ 
defined as in Section~\ref{sec:kappa}. It is easy to see that the operators 
$$\Dir_{g,\overline{g}}^{f'}:=\kappa_{g,\overline{g}}^{f'}\Dir_{\overline g}\kappa_{\overline{g},g}^{f'}:\Gamma(\S\M_g)\to \Gamma(\S\M_g) \quad \text{ and }\quad \Dir_{\overline{g},g_u}^{f''}:=\kappa_{g_u,\overline{g}}^{f''}\Dir_{\overline g}\kappa_{\overline{g},g_u}^{f''} : \Gamma(\S\M_{g_u})\to\Gamma(\S\M_{g_u})$$ are weakly-hyperbolic on $(\M,g)$ and $(\M,g_u)$ respectively, so we can construct a unitary M\o ller operator $\R_{g_u,g}:\sol(\Dir_{g})\to \sol(\Dir_{g_u})$, composing the unitary M\o ller operators $\R_{\overline g, g_u}:\sol(\Dir_{g_u})\to \sol(\Dir_{\overline g})$ and $\R_{g,\overline g}:\sol(\Dir_{\overline g})\to \sol(\Dir_{ g})$ obtained using the same arguments as in Sections~\ref{sec:Moller} and~\ref{sec:conserv}.
In particular, we can lift the action of the unitary M\o ller operator to a $*$-isomorphism between
 the algebra of Dirac fields on $(\M,g)$ and $(\M,g_u)$ respectively.
Hence for any Hadamard state $\omega_H$ on $\fA_u$, the state defined by
$$\omega=\omega_H \circ \mathfrak{R}_{1,0}\,: \fA \to \CC\,,$$
is also a Hadamard state on account of Theorem \ref{thm:main appl}.

It remains to show that there exists a Hadamard state $\omega_H$ for $\fA_u$.
For that, we shall define a quasi-free state by identifying a suitable operator $Q_\omega$ and then exploiting Equation \eqref{Eq: 2-point as an operator}.
In order to construct the desired $Q_\omega$, let us write the Dirac equation as $\Dir=\sigma(dt)\partial_t + \L$, where $\L$ differentiates only in the tangential part of $\Sigma$.
Since we coupled $\Dir$ with self-adjoint boundary condition, it follows that $\L$ is skew-adjoint.
As a consequence we may define the self-adjoint operator $\H=i\L$. To obtain a pure, quasi-free state it is enough to define the operator
$$Q_\omega:= P_+(\H) \oplus \Id_\mathcal{H} - \Upsilon P_+(\H) \Upsilon^{-1}$$
where  $P_+(\H)$ is the spectral projection in the positive spectrum of $\H$.
It is not difficult to see, that on globally hyperbolic ultrastatic manifolds with empty boundary, the associated quasi-free state is of the Hadamard form, since it provides the canonical frequencies splitting.
This concludes our proof.
\end{proof}

\begin{remark}
The main drawback of the definition of the M\o ller $*$-isomorphism $\mathfrak{R}$, used in Theorem~\ref{thm:main appl}, is the lack of any control on the action of the group of $*$-automorphism induced by the isometry group of $\M$ on $\omega_2$. Let us remark, that the study of invariant states is a well-established research topic (\cf \cite{NCTori,NCspinoff}). Indeed, the type of factor can be inferred by analyzing which and how many states are invariant. From a more physical perspective instead, invariant states can represent equilibrium states in statistical mechanics e.g. KMS-states or ground states.
\end{remark}
The previous remark leads us to the following open question: Under which conditions it is possible to perform an adiabatic limit, namely when is $ \lim\limits_{\chi\to 1} \omega_1$ well-defined? \medskip

A priori we expect that there is no positive answer in all possible scenarios, since
it is known that certain free-field theories, e.g., the massless and minimally coupled (scalar or Dirac) field on four-dimensional de Sitter spacetime, do not possess a ground state, even though their massive counterpart does.
Note that this is not a no-go Theorem, but at least an indication that, in these situation, the map $\omega\to\omega\circ\fR$ cannot be expected to preserve the ground state property.
A partial investigation in this direction has been carried on in \cite{Moller,Moller2} for the case of a scalar field theory on globally hyperbolic spacetimes with empty boundary. In this situation it has been shown that, under suitable hypotheses the adiabatic limit can be performed preserving the invariance property under time translation but spoiling in general the ground state or KMS property.
\medskip

Since our results depends only on the principal symbol of the Dirac operator and on the chosen boundary condition, we conclude our paper with the following corollary.

\begin{corollary}
Let $(\M,g)$ be a globally hyperbolic spin spacetime with timelike boundary and let $\Dir_{\mathcal{V}}=\Dir + \mathcal V$ be the Dirac operator coupled with a external skew-symmetric potential $\mathcal{V}\in\text{End}(\S\M)$ and with the MIT boundary condition.
Then there exists a  state for the algebra of Dirac fields with MIT boundary conditions which satisfies the Hadamard condition.
\end{corollary}

\end{document}